\documentclass{ectj}

\usepackage{amsfonts,amssymb,amsmath,bm,latexsym,url,mathtools,booktabs,multirow,xcolor,comment}
\usepackage{caption}
\usepackage[capitalise,noabbrev]{cleveref}

\newtheorem{theorem}{Theorem}[section]
\newtheorem{assumption}{Assumption}[section]

\newtheorem{corollary}{Corollary}[section]
\newtheorem{lemma}{Lemma}[section]
\newtheorem{definition}{Definition}[section]

\newtheorem{algorithm}{Algorithm}[section]
\numberwithin{equation}{section}
\crefname{assumption}{Assumption}{Assumptions}\Crefname{assumption}{Assumption}{Assumptions}
\crefname{definition}{Definition}{Definitions}\Crefname{definition}{Definition}{Definitions}
\crefname{algorithm}{Algorithm}{Algorithms}\Crefname{algorithm}{Algorithm}{Algorithms}

\renewcommand{\qed}{\hfill{\tiny \ensuremath{\blacksquare} }}%
\renewcommand{\Pr}{{\mathrm{P}}}

\newcommand{\E}{\mathrm{E}}

\makeatletter
\IfFileExists{accents.sty}{\usepackage{accents}\def\ubar##1{\underaccent{\bar}{##1}}}{\def\ubar##1{\underline{##1}}}
\makeatother

\title[Debiased inference for wage inequality]{Debiased Inference for Bounding Wage Inequality with Many Controls\thanks{The first version of this paper's theoretical results was posted on arXiv in August 2018 as arXiv:1808.02569, Chapter 4.}}

\author[Y. Korobka and V. Semenova]{Yaroslav~Korobka$^{\dagger}$ and Vira~Semenova$^{\ddagger}$}

\address{$^{\dagger}$CERGE-EI, Prague, Czech Republic.}
\email{yaroslav.korobka@cerge-ei.cz}

\address{$^{\ddagger}$Department of Economics, University of California, Berkeley, USA.}
\email{vsemenova@berkeley.edu}

\begin{document}

\begin{abstract}
We study estimation and inference for a partially identified parameter whose identified set depends on a first-stage nuisance parameter that must itself be estimated. 
Combining the criterion-function approach with the theory of Neyman-orthogonal moments that underlies double/debiased machine learning, we propose a two-step procedure: the point-identified nuisance is estimated by flexible machine-learning methods, and the set-identified target is recovered as a level set of a sample criterion built from orthogonal moment inequalities with cross-fitting. 
When the contour level is bounded, we show that the resulting set estimator converges in Hausdorff distance at the parametric rate of the infeasible criterion built on the true nuisance. We further develop a subsampling procedure that delivers asymptotically valid coverage, provided the product of the first-stage estimation errors is $o(N^{-1/2})$. We illustrate the method on bounds for the wage distribution and the interquantile range under selection into employment and on the gender wage gap with an interval-censored wage. The empirical application studies the gender wage gap using the March supplement of the 2015 Current Population Survey.
\end{abstract}

\section{Introduction}
\label{sec:intro}

Sample selection is a classic problem in econometrics \citep{Heckman74}. Economically important objects such as population wage distributions, gender and education wage differentials, and measures of wage inequality depend on the wages of individuals who do not work, and these wages are not observed. One approach imposes enough structure on the joint distribution of wages and employment to recover the latent wage distribution, for instance through a control function. For example, \citet{CFL} take this route with distribution regression to decompose British wages. An alternative approach leaves the selection process incompletely specified and asks what can be learned about the latent wage distribution under weaker restrictions, as in the partial identification analysis of \citet{Manski03}. In this paper, we focus on \citet{BGIM}, who take the second route to bound changes in the distribution of British wages. Their positive-selection restriction, under which the wage distribution of workers first-order stochastically dominates that of nonworkers conditional on observed characteristics, substantially tightens the worst-case bounds without point-identifying the missing wage distribution. The restriction is more plausible when the compared individuals are observationally similar, because it rules out selection only within narrow groups. Conditioning on many characteristics, however, creates a dimensionality problem.

This paper develops estimation and inference for partially identified models whose identifying inequalities depend on a high-dimensional first stage. The parameter of interest $\theta$ is set-identified by a finite collection of unconditional moment inequalities,
\begin{equation*}
    \mathrm{E}\,g(Z;\theta,\xi_0)\,\bm{\leq}\,0,
\end{equation*}
where the nuisance parameter $\xi_0$ is point-identified but may be infinite-dimensional. We estimate $\xi_0$ by machine-learning methods with cross-fitting, construct Neyman-orthogonal versions of the inequalities, and estimate the identified set as a level set of the criterion function of \citet{CHT}. We illustrate it with the wage-distribution bounds of \citet{BGIM} and a regression with an interval-valued outcome. The object of inference is the identified set with the confidence region covering $\Theta_I$ with probability approaching $1-\tau$, where $1-\tau$ is the nominal confidence level. 

The main technical contribution is subsampling inference that remains valid when the criterion is built from a cross-fitted first stage. Cross-fitting breaks the independence across subsamples that subsampling requires, because every block statistic uses a nuisance estimated on the whole sample and is maximized over a preliminary set built from the whole sample. We restore it using a carefully crafted balanced, disjoint-block construction together with a bracketing argument. Every block contains the same number of observations from each cross-fitting fold, so the observations it draws from fold $k$ are independent of the nuisance estimated without that fold. Each feasible block statistic is then bounded by statistics evaluated at the true nuisance on a small enlargement or contraction of the identified set. These bounding statistics depend only on the observations of their own block and can therefore be treated with ordinary block-level asymptotics. This yields consistency of the subsampling critical value and pointwise asymptotic coverage,
\begin{equation*}
    \Pr \left(\Theta_I \subseteq C_N (\hat{c}_\tau; \hat{\xi})\right) = 1 - \tau + o(1),
\end{equation*}
under $b \to \infty$ and $b c_N / N \to 0$, with no strengthening of the first-stage rate condition used for set estimation. Combining cross-fitting with dependent data raises a related difficulty, resolved by a cross-fitting scheme tailored to the dependence, as in the weakly dependent panels of \citet{CGST}, or by an additional stability condition, as in the network and spatial designs of \citet{CaoLeung}. 

A second contribution is a general theory that does not require the moment inequalities to be affine in $\theta$. Bounds on the interquantile range of the wage distribution under selection provide such an example. 

We demonstrate our method on the partially identified model studied in \citet{CCMS}, with wage data from the March supplement of the 2015 Current Population Survey. In particular, we document that, unlike an approach that reports only separate bounds on individual coefficients, the criterion-function formulation directly delivers a joint confidence region for the two-dimensional coefficient set. Across the designs we consider, the estimated joint identified set occupies only about $50\%$ of the box formed by its coordinate projections. We interpret this as evidence that separate coefficient bounds contain many combinations that are inconsistent with the joint directional restrictions.

This paper contributes to the literature on partial identification by incorporating a function-valued first stage into the two-stage framework of \citet{KaidoWhite}. In doing so, we benefit from the Neyman-orthogonality and cross-fitting machinery of \citet{chernozhukov2016double}, \citet{LRSP} and \citet{chernozhukov2021automatic}, so that a high-dimensional first stage can be estimated by regularized methods without affecting first-order inference. In partial identification, earlier applications of these techniques are mostly confined to bounds on one-dimensional parameters \citep{Heiler2024,Semenova2025} and to support functions of convex identified sets \citep{Semenova2023,LiuMolinari}, building on \citet{BeresteanuMolinari2008}, \citet{BMM2} and \citet{BontempsMagnacMaurin2012}. \citet{LiuMolinari}, for instance, propose debiased estimation and inference for an algorithmic fairness-accuracy frontier, and \citet{LiuSPT2025} develops a partial identification framework for panel-data policy evaluation that is robust to violations of the identifying assumptions of difference-in-differences and synthetic control. In contrast, our moment inequalities may be nonlinear in $\theta$ and describe potentially non-convex identified sets. Last but not least, the paper's emphasis on selection problems contributes to the sample selection literature, see \citet{FVVW2019,FVV2018,FVHong2025}.

Since the first version of this paper\footnote{\Cref{app:theory} and its proofs, except the extension to moments that are discontinuous in $\theta$, first appeared in Chapter~4 of arXiv:1808.02569, formerly titled ``Machine Learning for Dynamic Discrete Choice.''} (arXiv:1808.02569, August 2018), \citet{Tian2025} has studied a closely related two-stage problem. He combines the refined conditional chi-squared test of \citet{CoxShi2023} with a correction that propagates first-stage estimation uncertainty into the covariance matrix and inverts the test point by point, so that his confidence set covers each $\theta$ in the identified set with asymptotic probability at least the nominal level. The correction relies on an asymptotically linear first stage, such as GMM or maximum likelihood, and on a high-level condition that the plug-in moment vector be $\sqrt{N}$-asymptotically normal, which a machine-learning first stage that converges more slowly does not deliver without orthogonalization. We instead remove the first-order contribution of the estimated nuisance by orthogonalization, allow a first stage that converges more slowly than $N^{-1/2}$, and cover $\Theta_I$ as a whole. The two procedures thus treat the first stage differently: \citet{Tian2025} carries its first-order sampling variation into inference, while our construction makes that variation second order under suitable rate conditions. 

\paragraph{Structure of the paper.} \Cref{sec:setup} introduces the two-stage framework, the orthogonal criterion, and the estimation and inference procedure. \Cref{sec:example} presents the two examples: the wage distribution under selection of \citet{BGIM}, and the gender wage gap with an interval-valued outcome. \Cref{sec:examples-theory} gives primitive conditions and asymptotic guarantees for both examples, and \Cref{sec:empirical} reports the empirical application. Replication code is available at \texttt{https://github.com/korobkay/orthoset}. The appendix collects the general asymptotic theory, the proofs, and the verification of the high-level conditions for the two examples.

\section{Setup}
\label{sec:setup}
\subsection{Framework}
\label{sec:framework}

Let $Z_1,\ldots,Z_N$ be an i.i.d.\ sample from a distribution $P_0$ on a measurable space $\mathcal{Z}$. The \emph{target} parameter $\theta_0$ lies in a compact set $\Theta\subset\mathbb{R}^{d_\theta}$ and is the object of economic interest. The \emph{nuisance} parameter $\xi$ lies in a convex subset $\Xi$ of a normed space of functions of $Z$, with the $L^{2}(P_0)$ norm $\|\cdot\|_{P,2}$ unless an example specifies another norm. The nuisance has a true value $\xi_0\in\Xi$ that can be estimated without reference to $\theta_0$. Following \citet{chernozhukov2016double}, we assume that all objects are measurable and focus on estimation and inference. The pair $(\theta_0,\xi_0)$ has the two-stage structure of \citet{KaidoWhite}, with a point-identified first stage. Unlike  their setting, the first stage here is an (infinite-dimensional) nuisance function  rather than a Euclidean sub-vector.

A moment function $g:\mathcal{Z}\times\Theta\times\Xi\to\mathbb{R}^{L}$ defines the identified set through $L\geq 1$ inequalities evaluated at the true nuisance,
\begin{equation}\label{eq:idset}
    \Theta_I := \big\{\theta\in\Theta:\ \mathrm{E}\,g(Z;\theta,\xi_0)\,\bm{\leq}\,0\big\},
\end{equation}
where $v\,\bm{\leq}\,0$ for a vector $v\in\mathbb{R}^{L}$ means that every coordinate of $v$ is non-positive. Write $\|v\|_{+}:=\|\max\{v,0\}\|$ for the Euclidean norm of the positive part of $v$, the maximum being taken coordinatewise, so that $\|v\|_{+}=0$ if and only if $v\,\bm{\leq}\,0$. Define the \emph{population criterion}
\begin{equation}\label{eq:Qpop}
    Q(\theta,\xi) := \big\|\mathrm{E}\,g(Z;\theta,\xi)\big\|_{+}^{2}
\end{equation}
This criterion is non-negative for every $\theta\in\Theta$ and $\xi\in\Xi$. At the true value of the nuisance, it vanishes exactly on the identified set: $$\Theta_I=\{\theta\in\Theta: Q(\theta,\xi_0)=0\}.$$

We construct the sample criterion function by $K$-fold cross-fitting, with the number of folds $K\geq 2$ fixed. The indices $\{1,\ldots,N\}$ are partitioned into folds $J_1,\ldots,J_K$ of equal size $n=N/K$, and for $i\in J_k$ the nuisance $\widehat{\xi}_i:=\widehat{\xi}^{(k)}$ is estimated from the observations outside $J_k$ only. We write $\widehat{\xi}$ for the collection $(\widehat{\xi}^{(1)},\ldots,\widehat{\xi}^{(K)})$. The \emph{sample criterion} is
\begin{equation}\label{eq:Qn}
    Q_N(\theta,\widehat{\xi}) := \bigg\|\frac{1}{N}\sum_{i=1}^{N} g(Z_i;\theta,\widehat{\xi}_i)\bigg\|_{+}^{2},
\end{equation}
and for a deterministic $\xi\in\Xi$ we write $Q_N(\theta,\xi)$ for the same expression with $\xi$ in place of every $\widehat{\xi}_i$. At $\xi=\xi_0$ it is the criterion of \citet{CHT} for these moment inequalities. For a level $c\geq 0$ and a nuisance $\xi$, the contour set of the sample criterion is
\begin{equation}\label{eq:contour}
    C_N(c;\xi) := \big\{\theta\in\Theta:\ N\,Q_N(\theta,\xi)\leq c\big\}.
\end{equation}
\begin{definition}[Two-stage set estimator]
\label{def:setestim}
For a measurable level $\widehat{c}\geq 0$, the two-stage set estimator is
\begin{equation}\label{eq:setestim}
    \widehat{\Theta}_I := C_N(\widehat{c};\widehat{\xi}).
\end{equation}
\end{definition}
The level $\widehat{c}$ is chosen in one of two ways. A deterministic level $c_N\to\infty$ with $c_N/N\to 0$, for instance $c_N=\log N$, makes $\widehat{\Theta}_I$ a consistent set estimator. A data-driven level $\widehat{c}_\tau$, computed by subsampling, makes it a confidence region for $\Theta_I$ with asymptotic coverage $1-\tau$.

\begin{definition}[Subsampling critical value]
\label{def:subsampling}
Fix a deterministic level $c_N\to\infty$ with $c_N/N\to 0$ and a block size $b=b_N\to\infty$, a multiple of $K$, with $b\,c_N/N\to 0$. Partition the sample into $B_N:=\lfloor N/b\rfloor$ disjoint blocks $\mathcal{I}_1,\ldots,\mathcal{I}_{B_N}$ of size $b$, each containing $b/K$ observations from every fold, so that the observations of a block in fold $k$ are independent of $\widehat{\xi}^{(k)}$. For each block $j=1,\dots,B_N$ compute
\begin{equation*}
\widehat{\mathcal{T}}_{j,b} \;:=\; \sup_{\theta\in C_N(c_N;\widehat{\xi})} b\,Q_{j,b}(\theta,\widehat{\xi}),
\qquad
Q_{j,b}(\theta,\widehat{\xi}) := \bigg\|\frac{1}{b}\sum_{i\in\mathcal{I}_j} g(Z_i;\theta,\widehat{\xi}_i)\bigg\|_{+}^{2},
\end{equation*}
so that $\widehat{\mathcal{T}}_{j,b}$ is the criterion of block $\mathcal{I}_j$, scaled by $b$ and maximized over the preliminary set $C_N(c_N;\widehat{\xi})$. Let $\widehat{G}$ be the empirical distribution function of $(\widehat{\mathcal{T}}_{j,b})_{j=1}^{B_N}$. For $\tau\in(0,1)$, the subsampling critical value is the $(1-\tau)$-quantile
\begin{equation*}
\widehat{c}_{\tau} \;:=\; \inf\{x\in\mathbb{R}:\widehat{G}(x)\geq 1-\tau\}.
\end{equation*}
\end{definition}

Rates for the set estimator are stated in the Hausdorff metric,
\begin{equation}\label{eq:hausdorff}
    d_H(A,B) := \max\Big\{\sup_{x\in A}d(x,B),\ \sup_{y\in B}d(y,A)\Big\}, \qquad d(x,B):=\inf_{y\in B}\|x-y\|,
\end{equation}
and $\Theta_I^{\epsilon}:=\{\theta\in\Theta: d(\theta,\Theta_I)\leq\epsilon\}$ denotes the $\epsilon$-expansion of the identified set.

The first stage enters \eqref{eq:Qn} through $\widehat{\xi}$, whose estimation error is of larger order than $N^{-1/2}$ when $\xi_0$ is estimated by regularized methods. We require the moment function to be insensitive to that error to first order. For $\xi\in\Xi$ and $r\in[0,1]$ write $\xi_r:=\xi_0+r(\xi-\xi_0)$, which lies in $\Xi$.
\begin{definition}[Neyman orthogonality]
\label{def:ortho}
The moment function $g$ is \emph{Neyman-orthogonal} at $\xi_0$ on a set $\mathcal{N}\subset\Xi$ if, for every $\theta\in\Theta$ and $\xi\in\mathcal{N}$, the map $r\mapsto\mathrm{E}\,g(Z;\theta,\xi_r)$ is differentiable at $r=0$ with
\begin{equation}\label{eq:ortho}
    \partial_r\,\mathrm{E}\,g(Z;\theta,\xi_r)\big|_{r=0} = 0 .
\end{equation}
\end{definition}
Orthogonality says that, at the truth, the expected moment is flat in the direction of any first-stage error, so an error in $\widehat{\xi}$ moves $\mathrm{E}\,g$ only to second order.

\subsection{Proposed procedure}
\label{sec:overview}

The following algorithm combines the cross-fitted first stage, the criterion \eqref{eq:Qn} and the contour sets \eqref{eq:contour} of \Cref{sec:framework} with a subsampling choice of the level. The orthogonal moment function $g$ is an input, supplied case by case. Step~1 is the sample splitting of \citet{schick1986asymptotically}, in the cross-fitted form of \citet{chernozhukov2016double}, and Steps~2--4 are the contour-set estimator and the subsampling procedure of \citet{CHT} applied to the cross-fitted criterion. Step~3  incorporates the proposed delicate combination of subsampling with cross-fitting.

\begin{algorithm}[Two-stage set estimator and confidence region]
\label{alg:main}
Inputs: an orthogonal moment function $g$ with nuisance $\xi$, the number of folds $K\geq 2$, a preliminary level $c_N\to\infty$ with $c_N/N\to 0$, a block size $b\to\infty$, a multiple of $K$, with $b\,c_N/N\to 0$, and a nominal confidence level $1 - \tau$.
\begin{enumerate}
\item \emph{Cross-fitted first stage.} Partition $\{1,\ldots,N\}$ into folds $J_1,\ldots,J_K$ of equal size $n=N/K$. For each $k$, compute an estimate $\widehat{\xi}^{(k)}$ of $\xi_0$ from the observations outside $J_k$, and for $i\in J_k$ set $\widehat{\xi}_i:=\widehat{\xi}^{(k)}$.

\item \emph{Criterion and contour sets.} Form the cross-fitted criterion \eqref{eq:Qn} and the contour sets \eqref{eq:contour},
\begin{equation*}
    Q_N(\theta,\widehat{\xi}) = \bigg\|\frac{1}{N}\sum_{i=1}^{N} g(Z_i;\theta,\widehat{\xi}_i)\bigg\|_{+}^{2},
    \qquad
    C_N(c;\widehat{\xi}) = \big\{\theta\in\Theta:\ N\,Q_N(\theta,\widehat{\xi})\leq c\big\}.
\end{equation*}

\item \emph{Preliminary set and block statistics.} Form the preliminary set $C_N(c_N;\widehat{\xi})$ at the deterministic level $c_N$, for instance $c_N=\log N$. Partition the sample into $B_N=\lfloor N/b\rfloor$ disjoint blocks $\mathcal{I}_1,\ldots,\mathcal{I}_{B_N}$ of size $b$, each carrying $b/K$ observations from every fold. On block $j$ compute
\begin{equation*}
    Q_{j,b}(\theta,\widehat{\xi}) = \bigg\|\frac{1}{b}\sum_{i\in\mathcal{I}_j} g(Z_i;\theta,\widehat{\xi}_i)\bigg\|_{+}^{2},
    \qquad
    \widehat{\mathcal{T}}_{j,b} = \sup_{\theta\in C_N(c_N;\widehat{\xi})} b\,Q_{j,b}(\theta,\widehat{\xi}),
\end{equation*}

\item \emph{Level and region.} Let $\widehat{G}$ be the empirical distribution function of $(\widehat{\mathcal{T}}_{j,b})_{j\leq B_N}$ and let
\begin{equation*}
    \widehat{c}_{\tau} := \inf\big\{x:\ \widehat{G}(x)\geq 1-\tau\big\}
\end{equation*}
be its $(1-\tau)$-quantile, as in \Cref{def:subsampling}. Report the set estimator $\widehat{\Theta}_I=C_N(\widehat{c};\widehat{\xi})$ of \Cref{def:setestim} with $\widehat{c}=c_N$, or with $\widehat{c}=0$ when the sample criterion vanishes on a set with non-empty interior, together with the confidence region $C_N(\widehat{c}_{\tau};\widehat{\xi})$.
\end{enumerate}
\end{algorithm}

In applications, the supremum in Step~3 can be computed on a finite grid. The preliminary set is represented by the grid points at which $N\,Q_N(\theta,\widehat\xi)\leq c_N$, with spacing small relative to $\sqrt{c_N/N}$, and each block statistic is subsequently maximized over these points. When $g$ is affine in $\theta$, the block criterion and the preliminary set are convex, so only boundary grid points need to be evaluated.

\section{Examples}
\label{sec:example}

\subsection{Wage distribution under selection into employment}
\label{sec:example1}

Our first example is the wage-distribution problem of \citet{BGIM}. They study how the distribution of UK log hourly wages changed between 1978 and 2000 using the Family Expenditure Survey, and they report the interquartile range as a measure of within-group inequality together with educational, gender and cohort differentials in the median. Because wages are recorded only for active workers, these objects are only partially identified, and \citet{BGIM} bound them under restrictions on selection into work.

\subsubsection{Linear case: the distribution function on a grid}
\label{sec:example1-grid}

For each individual $i=1,\ldots,N$, let $W_i$ denote the log wage, $D_i\in\{0,1\}$ an employment indicator, and $X_i\in \mathbb{R}^{p}$ a vector of conditioning covariates. Selection into employment is non-ignorable, $W \not\perp D \mid X$, and we impose no exclusion restriction. Fix a grid of wage thresholds $w_1 < w_2 < \cdots < w_S$. The target parameter is the vector of population distribution-function values on the grid,
\begin{equation}\label{eq:gridtarget}
    \theta_0 := \big(F_0(w_1),\ldots,F_0(w_S)\big)', \quad F_0(w_k) := \Pr(W\leq w_k), \quad k = 1, \ldots, S.
\end{equation}

Write $F_0(w\mid x) := \Pr(W\leq w\mid X=x)$ for the conditional distribution function, $F_0 ^1(w\mid x)$ and $F_0 ^0(w\mid x)$ for its counterparts among the employed and the unemployed, and $p_0(x) := \Pr(D=1\mid X=x)$ for the \emph{employment propensity}. The law of iterated expectations gives
\begin{equation}\label{eq:lie}
    F_0(w\mid x) = F_0 ^1(w\mid x)\,p_0(x) + F_0 ^0(w\mid x)\,\big(1-p_0(x)\big).
\end{equation}
The propensity $p_0(x)$ and the employed-wage distribution $F_0 ^1 (w \mid x)$ are identified from the data (provided that $p_0(x) > 0$). On the contrary, the unemployed-wage distribution $F_0 ^0(w \mid x)$ is not identified, and it is restricted only by $F_0 ^0(w\mid x)\in[0,1]$. Substituting into \eqref{eq:lie} gives the worst-case bounds of \citet{Manski1994}, $$F_0 ^1(w\mid x)p_0(x)\leq F_0(w\mid x)\leq F_0 ^1(w\mid x)p_0(x)+1-p_0(x),$$ whose width $1-p_0(x)$ is the non-employment rate for a given value of $x$. \citet{BGIM} tighten the lower bound under positive selection into employment, that is, first-order stochastic dominance of the employed wage distribution, $F_0 ^1(w\mid x)\leq F_0 ^0(w\mid x)$. Under that restriction it holds that
\begin{equation}\label{eq:better_bounds}
    F_0 ^1(w \mid x) \leq F_0(w \mid x) \leq F_0 ^1 (w \mid x)\,p_0(x) + \big(1 - p_0(x)\big).
\end{equation}

To convert the problem to unconditional moment inequalities framework of \cref{sec:setup}, we integrate the bounds in \eqref{eq:better_bounds} over the distribution of $X$, resulting in the bounds for each coordinate of $\theta_0$,
\begin{equation}\label{eq:coordbounds}
    \ell_{0,k} := \mathrm{E}\big[F_0 ^1(w_k\mid X)\big], \quad
    u_{0,k} := \mathrm{E}\big[F_0 ^1(w_k\mid X)p_0(X) + 1 - p_0(X)\big],
\end{equation}
which gives
$$
\ell_{0,k} \leq \theta_{0,k} \leq u_{0,k}.
$$

A second restriction, that $F_0 ^0(\cdot\mid x)$ is a distribution function, ties the coordinates together. Applying \eqref{eq:lie} at $w_j<w_k$ and subtracting gives $F_0(w_k\mid x) - F_0(w_j\mid x) \geq p_0(x)(F_0 ^1 (w_k\mid x) - F_0 ^1(w_j\mid x))$, since $F_0^0(\cdot\mid x)$ is non-decreasing. Integrating over $X$ and using $$\mathrm{E}[p_0(X)(F_0 ^1(w_k\mid X)-F_0 ^1(w_j\mid X))] = \mathrm{E}[D\,\mathbf{1}\{w_j<W\leq w_k\}]$$ turns this into a restriction on the coordinates of $\theta_0$,
\begin{equation}\label{eq:increment}
    \theta_{0,k} - \theta_{0,j} \;\geq\; \mathrm{E}\big[D\,\mathbf{1}\{w_j<W\leq w_k\}\big] \;=\; u_{0,k} - u_{0,j}, \qquad j<k.
\end{equation}
Combining \eqref{eq:coordbounds} and \eqref{eq:increment} gives the identified set
\begin{equation}\label{eq:idsetgrid}
    \Theta_I = \Big\{\theta\in\Theta:\; \ell_{0, k}\leq\theta_k\leq u_{0, k} \;\;\forall k, \quad \theta_{k}-\theta_{k-1}\geq u_{0, k}-u_{0, k - 1}\;\;\forall k\geq 2\Big\},
\end{equation}
a convex polytope cut out by \(3S-1\) inequalities.

In the notation of \Cref{sec:setup}, the observation is $Z_i := (X_i, D_i, D_i W_i)$, $i = 1, \ldots, N$, and the moment function has $3S-1$ coordinates. For $k=1,\ldots,S$ the two coordinate moments are $g_{k}^{\mathrm{lo}}(Z;\theta,\xi) := \phi_L(Z;w_k,\xi) - \theta_k$ and $g_{k}^{\mathrm{up}}(Z;\theta) := \theta_k - \phi_U(Z;w_k)$, where
\begin{equation}\label{eq:scores}
\begin{aligned}
    \phi_L(Z;w,\xi) &:= F^1(w \mid X) + \frac{D\,\big(\mathbf{1}\{W \leq w\} - F^1(w \mid X)\big)}{p(X)}, \\
    \phi_U(Z;w) &:= D\,\mathbf{1}\{W \leq w\} + 1 - D,
\end{aligned}
\end{equation}
and for $k\geq 2$ the increment moments are
\begin{equation}\label{eq:incmoment}
    g_{k}^{\mathrm{inc}}(Z;\theta) := \mathbf{1}\{w_{k-1}<W\leq w_k\}D - (\theta_k-\theta_{k-1}).
\end{equation}
Collecting $\{g_k ^{\text{lo}}, g_k ^{\text{up}}\}_{k = 1} ^S$ and $\{g_k ^{\text{inc}}\}_{k = 2} ^{S}$ gives a moment function $g(Z;\theta,\xi)$ of length $L=3S-1$, and the identified set \eqref{eq:idset} for $\theta_0$ is \eqref{eq:idsetgrid}. The first-stage nuisance parameter is $\xi_0 := \big(p_0, F_0 ^1(w_1\mid \cdot),\ldots,F_0 ^1(w_S\mid \cdot)\big)$, estimated by cross-fitting.

Note that even in the linear case, it is necessary construct the orthogonal score $\phi_L$ as in \eqref{eq:scores}. The two plug-in estimators of the lower bound $\ell_{0,k}$, the regression form $N^{-1}\sum_{i}\widehat F^1(w_k\mid X_i)$ and the weighting form $N^{-1}\sum_{i}D_i\mathbf{1}\{W_i\leq w_k\}/\widehat p(X_i)$, inherit the first-stage error to first order. For a nuisance value $\xi=(p,F^1)$,
\begin{equation}\label{eq:pluginbias}
\begin{aligned}
    \mathrm{E}\big[F^1(w_k\mid X)\big]-\ell_{0,k} &= \mathrm{E}\big[F^1(w_k\mid X)-F_0^1(w_k\mid X)\big],\\
    \mathrm{E}\bigg[\frac{D\,\mathbf{1}\{W\leq w_k\}}{p(X)}\bigg]-\ell_{0,k} &= \mathrm{E}\bigg[F_0^1(w_k\mid X)\,\frac{p_0(X)-p(X)}{p(X)}\bigg],\\
    \mathrm{E}\,\phi_L(Z;w_k,\xi)-\ell_{0,k} &= \mathrm{E}\bigg[\big(F^1(w_k\mid X)-F_0^1(w_k\mid X)\big)\,\frac{p(X)-p_0(X)}{p(X)}\bigg].
\end{aligned}
\end{equation}
The first two errors are linear in a single first-stage error, typically of larger order than $N^{-1/2}$, so coverage of $\Theta_I$ is no longer guaranteed. The third error is the product of the two first-stage errors, which is $o(N^{-1/2})$ when each converges faster than $N^{-1/4}$. Every inequality in \eqref{eq:idsetgrid} is affine in $\theta$, so the orthogonal estimates of $\ell_{0,k}$ and $u_{0,k}$ could also be potentially used in the standard moment-inequality inference on the estimated system \citep{andrews2010inference,CoxShi2023}.

\subsubsection{Nonlinear case: the interquantile range}
\label{sec:example1-iqr}

The interquartile range that \citet{BGIM} report is a difference of two quantiles resulting in moment inequalities that are not affine in the parameter. Fix $0<\alpha_1<\alpha_2<1$, and suppose that the distribution of $W$ is continuous. The target parameter is the pair of population quantiles,
\begin{equation}\label{eq:iqrtarget}
    \theta_0 := \big(F_0^{-1}(\alpha_1),\,F_0^{-1}(\alpha_2)\big)', \qquad F_0(\theta_{0,j})=\alpha_j, \quad j=1,2,
\end{equation}
and the interquantile range is $\theta_{0,2}-\theta_{0,1}$. The parameter space is $\Theta:=\{\theta\in[\underline w,\bar w]^2:\theta_1\leq\theta_2\}$ for an interval $[\underline w,\bar w]$ that contains the support of $W$.

Write
\begin{equation*}
    F_0^{\mathrm{lo}}(w) := \mathrm{E}\big[F_0^1(w\mid X)\big], \qquad F_0^{\mathrm{up}}(w) := \mathrm{E}\big[F_0^1(w\mid X)\,p_0(X)+1-p_0(X)\big],
\end{equation*}
so that $\ell_{0,k}=F_0^{\mathrm{lo}}(w_k)$, $u_{0,k}=F_0^{\mathrm{up}}(w_k)$, and $F_0^{\mathrm{lo}}(w)\leq F_0(w)\leq F_0^{\mathrm{up}}(w)$ for every $w$. Evaluating these bounds at $w=\theta_{0,j}$, and the increment restriction \eqref{eq:increment} at the thresholds $\theta_{0,1}<\theta_{0,2}$, gives the identified set
\begin{equation}\label{eq:iqrset}
    \Theta_I=\Big\{\theta\in\Theta:\; F_0^{\mathrm{lo}}(\theta_j)\leq\alpha_j\leq F_0^{\mathrm{up}}(\theta_j),\;\; j=1,2, \quad F_0^{\mathrm{up}}(\theta_2)-F_0^{\mathrm{up}}(\theta_1)\leq\alpha_2-\alpha_1\Big\}.
\end{equation}
Suppose in addition that $\min\{\alpha_1,\alpha_2-\alpha_1\}>1-\Pr(D=1)$, $F_0^{\mathrm{lo}}$ and $F_0^{\mathrm{up}}$ are continuous, and strictly increasing near the points at which they equal $\alpha_1$ or $\alpha_2$. The coordinate projections of $\Theta_I$ are then the quantile bounds $[\underline\theta_j,\bar\theta_j]$, where $F_0^{\mathrm{up}}(\underline\theta_j)=\alpha_j$ and $F_0^{\mathrm{lo}}(\bar\theta_j)=\alpha_j$.

The increment restriction can make the joint bound on the interquantile range strictly tighter than the difference of the two coordinate bounds. Since
\[
F_0^{\mathrm{up}}(w)-F_0^{\mathrm{lo}}(w)
=
\mathrm{E}[(1-p_0(X))(1-F_0^1(w\mid X))],
\]the box corner \((\underline\theta_1,\bar\theta_2)\), which uniquely maximizes \(\theta_2-\theta_1\) over the Cartesian product of the coordinate bounds, satisfies
\[
F_0^{\mathrm{up}}(\bar\theta_2)
-
F_0^{\mathrm{up}}(\underline\theta_1)
=
\alpha_2-\alpha_1+
\mathrm{E}[(1-p_0(X))(1-F_0^1(\bar\theta_2\mid X))].
\]Hence this corner is excluded whenever the last expectation is positive. Since \(\Theta_I\) is compact, the upper bound on the interquantile range is then strictly smaller than \(\bar\theta_2-\underline\theta_1\). This is the unconditional analogue of the cdf coherence argument used by \citet[Section~2.1.3]{BGIM} to tighten bounds on conditional within-group inequality.

In the notation of \Cref{sec:setup}, the observation is $Z_i:=(X_i,D_i,D_iW_i)$, and the moment function has five coordinates, obtained by evaluating the scores \eqref{eq:scores} at the parameter,
\begin{equation}\label{eq:iqrmoments}
\begin{aligned}
    g_j^{\mathrm{lo}}(Z;\theta,\xi) &:= \phi_L(Z;\theta_j,\xi)-\alpha_j, \qquad g_j^{\mathrm{up}}(Z;\theta):=\alpha_j-\phi_U(Z;\theta_j), \qquad j=1,2,\\
    g^{\mathrm{inc}}(Z;\theta) &:= D\,\mathbf{1}\{\theta_1<W\leq\theta_2\}-(\alpha_2-\alpha_1).
\end{aligned}
\end{equation}
Because
\[
\mathrm{E}\,\phi_L(Z;w,\xi_0)=F_0^{\mathrm{lo}}(w),\qquad
\mathrm{E}\,\phi_U(Z;w)=F_0^{\mathrm{up}}(w),
\]and, for \(\theta_1\leq\theta_2\),
\[
\mathrm{E}[D\,\mathbf{1}\{\theta_1<W\leq\theta_2\}]
=
F_0^{\mathrm{up}}(\theta_2)-F_0^{\mathrm{up}}(\theta_1),
\]the zero set of the population criterion generated by \eqref{eq:iqrmoments} is \eqref{eq:iqrset}. The nuisance parameter is \(\xi_0=(p_0,F_0^1)\). The bias identity in \eqref{eq:pluginbias} holds at every threshold \(w\), so the lower-bound moments are Neyman-orthogonal pointwise in \(\theta\); \Cref{ass:iqrrates} supplies the uniform control over \(\theta\) required by the asymptotic theory.

Three features of \eqref{eq:iqrmoments} distinguish the interquantile case from the linear case. First, the moments are non-affine in \(\theta\), because the quantile locations enter through both the indicator functions and \(F^1(\theta_j\mid X)\). Second, the nuisance is evaluated at the parameter, so \(F_0^1(w\mid x)\) must be estimated uniformly over \(w\). Third, \(\Theta_I\) need not be convex: the increment restriction defines a nonlinear upper boundary for \(\theta_2\) as a function of \(\theta_1\), which need not be concave. Recovering quantiles from the fixed-threshold system of \Cref{sec:example1-grid} would instead require a threshold grid whose mesh tends to zero, and therefore a growing number of inequalities. Alternatively, one could estimate \(F_0^{\mathrm{lo}}\) and \(F_0^{\mathrm{up}}\) uniformly and invert the estimated bound functions. The criterion approach avoids this separate inversion step: the quantile locations enter directly as parameters, and the zero set of the criterion yields \(\Theta_I\) even when that set is nonconvex.

\subsection{Wage gap with an interval-censored outcome}
\label{sec:example2}

Our second example is the gender wage gap when the wage is recorded only as a bracket, a classic source of partial identification \citep{ManskiTamer02}. We describe the set of coefficients consistent with the brackets by a family of moment inequalities.

The target is the coefficient on the treatment in a partially linear projection of the log wage, of which only a bracket is observed. The projection is motivated by the model
\begin{equation}\label{eq:plm-emp}
    Y = D'\theta_0 + f_0(X) + U, \qquad \mathrm{E}[U\mid X, D] = 0,
\end{equation}
where $D\in\mathbb{R}^{d}$ is a low-dimensional treatment vector, $f_0$ is an integrable function, and $X \in \mathbb{R}^p$ is a vector of controls whose dimension may be large ($p \gg N$). The target parameter is the partial-linear projection coefficient $\theta_0:=\Sigma_0^{-1}\mathrm{E}[V(\eta_0)Y]$, with $V(\eta_0)$ and $\Sigma_0$ as in \eqref{eq:idset-emp}. It equals the coefficient in \eqref{eq:plm-emp} when the model holds and remains well defined when effects are heterogeneous. The latent value of $Y$ is bracketed by the observable bounds,
\begin{equation}\label{eq:bracket}
    Y_L \;\leq\; Y \;\leq\; Y_U, \qquad \Delta := Y_U-Y_L,
\end{equation}
where the width $\Delta$ can be random.

The identified set is the set of partial-linear projection coefficients generated by outcomes consistent with the brackets,
\begin{equation}\label{eq:idset-emp}
    \Theta_I^{\ast} \;:=\; \big\{\Sigma_0^{-1}\mathrm{E}[V(\eta_0) Y]\;:\; Y_L\leq Y\leq Y_U \ \text{a.s.}\big\},
\end{equation}
where $\eta_0(X) := \E[D \mid X]$, $V(\eta) := D - \eta(X)$ is the residualized treatment and $\Sigma_0 := \E [ V(\eta_0)V(\eta_0)']$. Let $\ell_0(X):=\mathrm{E}[Y_L\mid X]$ and $U_L := Y_L-\ell_0(X)$ be the residualized lower bracket. By definition, $\Sigma_0\theta_0=\mathrm{E}[V(\eta_0) Y]$, and decomposing $Y=Y_L+R$ with $R:=Y-Y_L$ splits this into an observed and an unobserved part,
\begin{equation}\label{eq:split}
    \Sigma_0\theta_0 \;=\; \mathrm{E}[V(\eta_0) U_L] \;+\; \mathrm{E}[V(\eta_0) R],
\end{equation}
where the first term uses $\mathrm{E}[V(\eta_0)\mid X]=0$.

Bounding the unobserved term produces one supporting half-space of $\Theta_I^{\ast}$ for every direction $q$. Since $0\leq R\leq\Delta$, the bound $(q'V(\eta_0))R\leq\Delta\,(q'V(\eta_0))_{+}$ holds pointwise for every $q$, where $a_+ := \max\{0, a\}$, and taking expectations gives
\begin{equation}\label{eq:halfspace}
    q'\Sigma_0\theta_0 \;\leq\; \mathrm{E}[(q'V(\eta_0))U_L] + \mathrm{E}\big[\Delta\,(q'V(\eta_0))_{+}\big].
\end{equation}
The set $\Theta_I ^*$ is convex, so it is the intersection of its supporting half-spaces: \eqref{eq:halfspace} is the supporting half-space in direction $q$, and intersecting over all $q$ on the unit sphere recovers $\Theta_I ^*$. For implementation we use a finite grid $q_1, \ldots, q_L$ and let $\Theta_I$ denote the resulting polyhedral outer approximation,
\begin{equation}\label{eq:ex2finite_set}
    \Theta_I = \left\{\theta \in \Theta: q_l' \Sigma_0 \theta \leq \mathrm{E}[(q_l'V(\eta_0))U_L] + \mathrm{E}\big[\Delta\,(q_l'V(\eta_0))_{+}\big], \quad l = 1, \ldots, L \right\}.
\end{equation}
Each direction gives one moment inequality, with its own orthogonalization because its nuisance enters through $q'V(\eta)$.

In the notation of \Cref{sec:setup}, the observation is $Z_i :=(X_i, D_i, Y_{L, i}, Y_{U, i})$, $i = 1, \ldots, N$. Consider recentering the positive-part term in \eqref{eq:halfspace} by $c(X) (q' V(\eta))$ for a measurable function $c$, which leaves the expectation at $\eta_0$ unchanged because $\mathrm{E}[V(\eta_0)\mid X]=0$. Whenever the positive-part expectation is pathwise differentiable at $\eta_0$, perturbing $\eta$ in the direction $\delta$ gives the derivative at $\eta_0$,
\begin{equation}\label{eq:piq}
    -\mathrm{E}[(q'\delta(X))\{\pi_q(\eta_0, X)-c(X)\}], \quad \pi_q(\eta, X) \;:=\; \mathrm{E}\big[\Delta\,\mathbf{1}\{q'V(\eta)>0\}\mid X\big],
\end{equation}
which vanishes along every direction $\delta$ at $c(X) = \pi_q(\eta_0, X)$. The orthogonal moment function is therefore
\begin{equation}
    \label{eq:emp-moment}
    g_q(Z;\theta,\xi) := (q'V(\eta))(V(\eta)'\theta) - (q'V(\eta))U_L - \Big[\Delta\,(q'V(\eta))_{+} - (q'V(\eta))\,\pi_q(\eta, X)\Big].
\end{equation}
Collecting $g_1, \ldots, g_L$ gives a moment function $g(Z; \theta, \xi)$ of length $L$, and the identified set \eqref{eq:idset} for $\theta_0$ is \eqref{eq:ex2finite_set}. The first-stage nuisance parameter is $\xi_0 := (\eta_0, \ell_0, \pi_{0,1}, \ldots, \pi_{0,L})$, estimated by cross-fitting.

\section{Formal results}
\label{sec:examples-theory}

The three corollaries of this section apply \Cref{alg:main} to the moment functions of \Cref{sec:example} and state, under primitive conditions, the rate of the set estimator and the coverage of the confidence region.

\subsection{Wage distribution under selection into employment}\label{sec:formal_bgim}

\Cref{ass:ex1id} collects the restrictions that define the identified set \eqref{eq:idsetgrid}. \Cref{ass:ex1ovl} is the strict overlap that the inverse propensity in \eqref{eq:scores} requires. \Cref{ass:ex1rates} places standard restrictions on the convergence rates of the nuisance parameter estimators.

\begin{assumption}[Identification]
\label{ass:ex1id}
(1) For $j=1,\ldots,S$, the employed-wage distribution stochastically dominates the unemployed-wage distribution, that is, $F_0 ^1(w_j\mid X)\leq F_0 ^0(w_j\mid X)$ almost surely. (2) The parameter space is $\Theta=[0,1]^S$; under (1) the polytope \eqref{eq:idsetgrid} contains $\theta_0$ and is non-empty.
\end{assumption}

\begin{assumption}[Overlap]
\label{ass:ex1ovl}
There exists $\underline p>0$ such that $p_0(X)\geq\underline p$ almost surely and every fold estimate satisfies $\widehat p^{(k)}(X)\geq\underline p$ almost surely, $k=1,\ldots,K$.
\end{assumption}

\begin{assumption}[Convergence Rates]
\label{ass:ex1rates}
There exist sequences $g_N^{p}$ and $g_N^{F}$ such that, with probability approaching one, $\|\widehat p^{(k)}-p_0\|_{P,2}\leq g_N^{p}$ and $\max_{j\leq S}\|\widehat F^{1(k)}(w_j\mid\cdot)-F_0 ^1(w_j\mid\cdot)\|_{P,2}\leq g_N^{F}$ for every fold $k$, and
\begin{equation*}
    g_N^{p}\,g_N^{F} = o(N^{-1/2}), \qquad g_N^{p}\vee g_N^{F} = o\big((\log N)^{-1/2}\big).
\end{equation*}
\end{assumption}

\begin{corollary}[Wage Distribution under Selection]
\label{cor:ex1}
Suppose \Cref{ass:ex1id,ass:ex1ovl,ass:ex1rates} hold, and let $g$ collect the $2S$ coordinate moments built from \eqref{eq:scores} and the $S-1$ increment moments \eqref{eq:incmoment}. Let $\widehat\Theta_I=C_N(c_N;\widehat\xi)$ be the set estimator of \Cref{def:setestim} at the level $c_N=\log N$. Then the following hold:
\begin{enumerate}
\item[(1)] $\Pr(\Theta_I\subseteq\widehat\Theta_I)\rightarrow 1$ and $d_H(\widehat\Theta_I,\Theta_I)=O_P\big(\sqrt{\log N/N}\big)$.
\item[(2)] Let $\tau\in(0,1)$ be such that the limit law of $\sup_{\theta\in\Theta_I}N\,Q_N(\theta,\xi_0)$ puts mass less than $1-\tau$ at zero and has a distribution function that is strictly increasing in a neighbourhood of its $(1-\tau)$-quantile, and let $b\rightarrow\infty$ with $b\log N/N\rightarrow 0$. The critical value $\widehat c_\tau$ of \Cref{alg:main} then satisfies $$\Pr\big(\Theta_I\subseteq C_N(\widehat c_\tau;\widehat\xi)\big)= 1-\tau+o(1).$$
\item[(3)] Suppose $\Theta_I$ has non-empty interior. Then $d_H\big(C_N(\widehat c';\widehat\xi),\Theta_I\big)=O_P(N^{-1/2})$ for every level $\widehat c'=O_P(1)$ that is at least $N\min_{\theta\in\Theta}Q_N(\theta,\widehat\xi)$ with probability approaching one. The level $\widehat c'=0$ is admissible whenever the sample criterion attains zero with probability approaching one.
\end{enumerate}
\end{corollary}

\Cref{cor:ex1} gives the asymptotic theory for the set estimator and the confidence region in the wage-distribution example. The rate in (1) and the coverage in (2) match the infeasible criterion $Q_N(\theta,\xi_0)$ of \citet{CHT}, and (3) is their degenerate case. The first stage converges more slowly than $N^{-1/2}$ but does not affect the first-order asymptotics: \Cref{ass:ex1rates} restricts only the product of the two first-stage errors, so an $o(N^{-1/4})$ rate for each suffices.

\subsection{Interquantile range under selection into employment}\label{sec:formal_iqr}

\Cref{ass:iqrid} extends the stochastic-dominance restriction of \Cref{ass:ex1id} to every threshold and requires the upper bound function $F_0^{\mathrm{up}}$ of \Cref{sec:example1-iqr} to increase at a positive rate near the quantile bounds. \Cref{ass:iqrrates} restates \Cref{ass:ex1rates} uniformly over thresholds. Distribution regression, which estimates $F_0^1(w\mid x)$ at each threshold $w$ by an $\ell_1$-penalized logistic regression of $\mathbf{1}\{W\leq w\}$ on $X$ among the employed, attains rates of this type uniformly over thresholds \citep{Program}. Its fitted values need not be monotone in $w$. Rearranging them in $w$ restores monotonicity and does not increase $\sup_w|\widehat F^{1(k)}(w\mid x)-F_0^1(w\mid x)|$ at any $x$ \citep{CFG2010}, so it preserves the rate condition whenever that condition holds in the stronger norm $\|\sup_w|\cdot|\|_{P,2}$.

\begin{assumption}[Identification]
\label{ass:iqrid}
(1) For every $w$, the employed-wage distribution stochastically dominates the unemployed-wage distribution, that is, $F_0^1(w\mid X)\leq F_0^0(w\mid X)$ almost surely. (2) The distribution of $W$ is continuous, its support lies in $[\underline w,\bar w]$, and given $D=1$ and $X$ the wage has a density $f_0^1(w\mid X)$ bounded by a constant $\bar f$. (3) The quantile levels satisfy $\min\{\alpha_1,\alpha_2-\alpha_1\}>1-\Pr(D=1)$, and there exist $\bar\delta>0$ and $\underline f>0$ such that $\mathcal{J}:=[\underline\theta_1-\bar\delta,\bar\theta_2+\bar\delta]\subset(\underline w,\bar w)$ and $\mathrm{E}[p_0(X)f_0^1(w\mid X)]\geq\underline f$ for almost every $w\in\mathcal{J}$.
\end{assumption}

\begin{assumption}[Convergence Rates]
\label{ass:iqrrates}
For every fold $k$ and every $x$, the estimate $\widehat F^{1(k)}(\cdot\mid x)$ is non-decreasing on $[\underline w,\bar w]$ with values in $[0,1]$. There exist sequences $g_N^{p}$ and $g_N^{F}$ that satisfy the rate conditions of \Cref{ass:ex1rates} and such that, with probability approaching one, $\|\widehat p^{(k)}-p_0\|_{P,2}\leq g_N^{p}$ and $\sup_{w\in[\underline w,\bar w]}\|\widehat F^{1(k)}(w\mid\cdot)-F_0^1(w\mid\cdot)\|_{P,2}\leq g_N^{F}$ for every fold $k$.
\end{assumption}

\begin{corollary}[Interquantile Range under Selection]
\label{cor:iqr}
Suppose \Cref{ass:ex1ovl,ass:iqrid,ass:iqrrates} hold, and let $g$ collect the five moments in \eqref{eq:iqrmoments}. Then conclusions (1), (2) and (3) of \Cref{cor:ex1} hold.
\end{corollary}

\Cref{cor:iqr} shows that the conclusions of \Cref{cor:ex1} continue to hold for the non-affine moment system \eqref{eq:iqrmoments}. Because the lower-bound moments evaluate \(\widehat F^1\) at the unknown quantile locations, the first-stage rate must hold uniformly over thresholds. The lower density bound in \Cref{ass:iqrid}(3), together with the upper density bound in \Cref{ass:iqrid}(2), provides the local linear error bound that is automatic for the affine system of \Cref{cor:ex1}. In particular, the transformation
\[
\Phi(\theta)
=
\big(F_0^{\mathrm{up}}(\theta_1),
F_0^{\mathrm{up}}(\theta_2)\big)
\]is locally bi-Lipschitz on the relevant region, and the transformed identified set is a convex polygon, so Hoffman's bound yields a population moment violation proportional to \(d(\theta,\Theta_I)\). Let $r(\theta) := \theta_2 - \theta_1$ and $\mathcal{R}_I := r(\Theta_I)$. Although $\Theta_I$ need not be convex, it is nevertheless connected. Specifically, under the transformation $\Phi$ used in the proof of \cref{cor:iqr}, it is the inverse image of a convex polygon, and as a result, $\mathcal{R}_I$ is an interval. Since $r$ is Lipschitz, $d_H (r(\widehat{\Theta}_I), \mathcal{R}_I) \leq \sqrt{2} d_H (\widehat{\Theta}_I, \Theta_I)$, so that the estimated interquantile-range set converges at no slower rate than the rate in \cref{cor:iqr}(1). A confidence interval for the identified interquartile range is obtained from the interval hull
\begin{equation*}
    \text{CI}_\tau := \left(\inf _{\theta \in C_N(\hat{c}_\tau; \hat{\xi})} r(\theta), \sup_{\theta \in C_N(\hat{c}_\tau; \hat{\xi})} r(\theta) \right).
\end{equation*}
Whenever $C_N(\hat{c}_\tau; \hat{\xi})$ contains $\Theta_I$, $\text{CI}_\tau$ contains $\mathcal{R}_I$. Therefore, $\lim \inf _{N \rightarrow \infty} \Pr (\mathcal{R}_I \subseteq \text{CI}_\tau) \geq 1 - \tau$. This projection is generally conservative for the scalar target because $\hat{c}_\tau$ is calibrated to cover the entire set $\Theta_I$. By contrast, \citet{KaidoMolinariStoye} calibrate inference directly for a component or smooth function of $\theta$.

\subsection{Wage gap with an interval-censored outcome}\label{sec:formal_vira}

\Cref{ass:ex2id} imposes a standard identification condition and ensures that the finite system of directional inequalities defines a bounded polyhedron. \Cref{ass:ex2width} provides bounded envelopes for the moments. The only condition specific to the non-smooth positive-part term is \Cref{ass:ex2kink}, which requires its population remainder to be quadratic in the propensity error. Finally, \Cref{ass:ex2rates} requires the resulting first-stage bias to be \(o(N^{-1/2})\) and the \(L^2\) nuisance errors to satisfy the continuity rate.

\begin{assumption}[Identification]
\label{ass:ex2id}
(1) The matrix $\Sigma_0 := \mathrm{E}[(D - \eta_0(X))(D - \eta_0(X))']$ is non-singular. (2) The finite set $q_1,\ldots,q_L \in \mathbb{S}^{d - 1}$ is symmetric about the origin, and spans $\mathbb{R}^d$. (3) $\Theta$ is compact and contains the polyhedron cut out by the $L$ half-spaces \eqref{eq:halfspace}.
\end{assumption}

\begin{assumption}[Boundedness]
\label{ass:ex2width}
For finite constants $\bar{D}$, $\bar{Y}$, $\bar{\Delta}$ it holds that $\|D\| \leq \bar{D}$, $|Y_L| \leq \bar{Y}$ and $\Delta \leq \bar{\Delta}$ almost surely.
\end{assumption}

\begin{assumption}[Quadratic kink remainder]
    \label{ass:ex2kink}
    For every $l = 1, \ldots, L$ there exists $C_K < \infty$ such that for every $\eta$ with $\|\eta(X)\|\leq\bar{D}$ almost surely,
    \begin{equation}
        |\mathcal{K}_l (\eta)| \leq C_K \|\eta - \eta_0\|_{P, 2} ^2,
    \end{equation}
    where
    \begin{equation*}
        \mathcal{K}_l (\eta) := \mathrm{E}\left[\Delta ((q_l' V(\eta))_+ - (q_l' V(\eta_0))_+ + q_l' (\eta - \eta_0)(X) \mathbf{1} \{q_l' V(\eta_0) > 0\})\right].
    \end{equation*}
\end{assumption}

\begin{assumption}[Convergence rates]
    \label{ass:ex2rates}
    Every fold estimate satisfies $\|\hat{\eta}^{(k)}(X)\|\leq\bar{D}$, $|\hat{\ell}^{(k)}(X)|\leq\bar{Y}$ and $0\leq\hat{\pi}_{l}^{(k)}(X)\leq\bar{\Delta}$ almost surely, $l=1,\ldots,L$. There exist sequences $g_{\eta, N}$, $g_{\ell, N}$, $g_{\pi, N}$ such that with probability approaching one for every fold $k$, $\|\hat{\eta}^{(k)} - \eta_0\|_{P, 2} \leq g_{\eta, N}$,  $\|\hat{\ell}^{(k)} - \ell_0\|_{P, 2} \leq g_{\ell, N}$, $\max_{l \leq L} \|\hat{\pi}_{l}^{(k)} - \pi_{0, l}\|_{P, 2} \leq g_{\pi, N}$, and
    \begin{equation*}
        g_{\eta, N} (g_{\eta, N} \vee g_{\ell, N} \vee g_{\pi, N}) = o(N^{-1/2}), \quad g_{\eta, N} \vee g_{\ell, N} \vee g_{\pi, N} = o((\log N)^{-1/2}).
    \end{equation*}
\end{assumption}

\begin{corollary}[Wage Gap with an Interval-Censored Wage]
\label{cor:ex2}
Suppose \Cref{ass:ex2id,ass:ex2width,ass:ex2kink,ass:ex2rates} hold, and let $g=(g_1,\ldots,g_L)$ be built from \eqref{eq:emp-moment}. Then conclusions (1), (2) and (3) of \Cref{cor:ex1} hold.
\end{corollary}

The role of \cref{ass:ex2kink} is to control observations near the kink $q' V(\eta_0) = 0$. It can be verified in different ways depending on the design. For continuous designs, a conventional margin or anti-concentration condition around zero implies the quadratic bound. In discrete designs, the same condition may instead hold through sign stability: if perturbing $\eta$ does not change the sign of $q' V(\eta)$, then the kink remainder is exactly zero. The empirical application of \cref{sec:empirical} is a particularly simple special case. Conditional on $X$, $D$ has two support points, and the estimated propensity remains in \((0,1)\). Perturbing the propensity therefore moves the two residual support points without changing their signs, so the kink remainder vanishes exactly. Moreover, with constant bracket width $\Delta$, each $\pi_l$ is available as a closed-form function of the scalar propensity. Consequently, the general rate condition stated in \cref{ass:ex2rates} reduces to $g_{\eta, N} (g_{\eta, N} + g_{\ell, N}) = o(N^{-1/2})$. 

\section{Empirical application}
\label{sec:empirical}

We use the interval-augmented March Supplement of the 2015 Current Population Survey data where we intentionally bracket an observed outcome (log wage) to demonstrate our procedure. The sample is restricted to white non-Hispanic individuals aged 25--64 who worked more than 35 hours per week for at least 50 weeks, resulting in a total of $N = 32{,}523$ observations. The outcome $Y$ is the log hourly wage, and the treatment vector is $D = (\text{female}, \text{female}\times \text{college})'$. Therefore, $\theta_1$ measures the gender wage gap among non-college workers, and $\theta_2$ measures the additional gender gap differential associated with college. The conditioning vector $X$ contains 259 demographic, region, and experience controls and their interactions. We construct the interval-censored wage while retaining the observed wage as an infeasible benchmark. For each bracket width $\Delta \in \{1, 2, 3\}$, we replace $Y$ by lower and upper endpoints $Y_L$ and $Y_U$ satisfying $\Delta = Y_U - Y_L$. Thus the information available to the estimator is only that $Y \in [Y_L, Y_U]$, while observed $Y$ is retained for the subsequent analysis.

For each $\Delta$, we apply the procedure of Sections \ref{sec:example2} and \ref{sec:overview}. The first-stage nuisance functions are the propensity $\eta_{0, 1} (X) = \Pr(\text{female} = 1 \mid X)$, and the conditional mean of the lower bracket, $\ell_0 (X) = \E[Y_L \mid X]$. Both nuisance components are estimated by $\ell_1$-penalized regressions with $K=2$ fold cross-fitting and the penalty selection rule of \cite{Program}. Note that the second component of the treatment propensity satisfies $\eta_{0,2} (X) = \eta_{0,1} (X) \cdot \text{college}$. We impose the orthogonal moment inequality \eqref{eq:emp-moment} at $L = 180$ equally spaced directions on the unit circle. For inference, we form the preliminary contour at $c_N = \log N$ and calibrate the final level by the balanced-block subsampling procedure of \cref{def:subsampling}. We set $b = 160$, giving $B_N = 203$ disjoint blocks, with each block drawing the same number of observations from each cross-fitting fold. The reported confidence region uses the 0.95 empirical quantile of the block statistics.

We contrast our method with two other procedures: the estimated set based on the non-orthogonal moment function and the infeasible regression using the observed wage. \cref{tab:cps} reports the results along the estimated identified set based on the orthogonal moment function and its 95\% confidence region. The comparison between the orthogonal and non-orthogonal set estimators reveals only a marginal difference. In this application, the numerical difference is small with at most 0.003$\Delta$ difference in each coordinate projection across all bracket widths. The negligible difference can be explained by the sufficiently accurate estimates of the propensity scores in the first step. Finally, a point-estimate resulting from the infeasible regression is well-covered by the estimated sets.

\cref{fig:cps} shows the estimated sets for the three bracket widths. First, it reflects how loss of information created by interval censoring influences the size of the estimated sets. Larger brackets necessarily result in larger estimated sets. Second, it describes how beneficial joint set estimation is compared to its coordinatewise projections imposed independently. Across all three values of $\Delta$, the area of the estimated set $\widehat{\Theta}_I$ is approximately one half of the area of the set resulting from the coordinatewise independent bounds. This discrepancy is a result of dependence between the two coefficient restrictions that is discarded when restrictions are imposed independently.

Two qualifications are important. First, the directional inequalities characterize the outer identified set associated with the restrictions developed in \cref{sec:example2}; as discussed in \cite{Semenova2023}, this set need not be sharp for the underlying structural coefficient $\theta_0$. Second, implementation replaces the continuum of directions by $L = 180$ directions, so the reported estimated set is itself a finite-direction outer approximation. Increasing the number of directions might therefore produce better approximations of the outer set.

\begin{figure}[htbp]
\centering
\includegraphics[width=\textwidth]{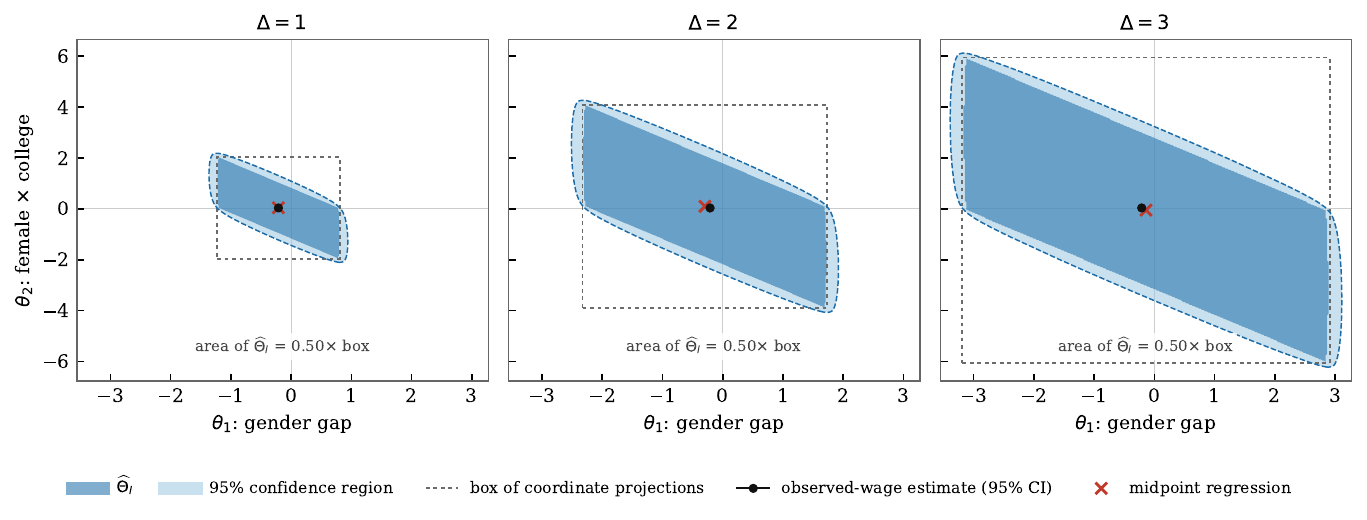}
\caption{\label{fig:cps}Estimated identified set and $95\%$ confidence region for the bracket widths $\Delta\in\{1, 2, 3\}$.}
\end{figure}

\begin{table}[htbp]
\centering
\small
\caption{\label{tab:cps}Bounds on the gender wage gap and its college interaction with a bracketed log wage}
\begin{tabular}{clcccc}
\toprule
$\Delta$ & & $\theta_1$ (gender gap) & $\theta_2$ (female $\times$ college) & area\,/\,box & $\widehat{c}_{\tau}$ \\
\midrule
\multirow{3}{*}{$1$}
 & $\widehat\Theta_I$               & $[-1.227,\ 0.805]$ & $[-1.960,\ 2.040]$ & $0.496$ & \multirow{3}{*}{$16.61$} \\
 & $95\%$ region                    & $(-1.363,\ 0.941)$ & $(-2.103,\ 2.184)$ & & \\
 & uncorrected moment               & $[-1.229,\ 0.807]$ & $[-1.957,\ 2.037]$ & & \\
\midrule
\multirow{3}{*}{$2$}
 & $\widehat\Theta_I$               & $[-2.327,\ 1.736]$ & $[-3.903,\ 4.097]$ & $0.496$ & \multirow{3}{*}{$23.41$} \\
 & $95\%$ region                    & $(-2.511,\ 1.920)$ & $(-4.073,\ 4.268)$ & & \\
 & uncorrected moment               & $[-2.332,\ 1.741]$ & $[-3.896,\ 4.091]$ & & \\
\midrule
\multirow{3}{*}{$3$}
 & $\widehat\Theta_I$               & $[-3.187,\ 2.909]$ & $[-6.050,\ 5.950]$ & $0.496$ & \multirow{3}{*}{$24.31$} \\
 & $95\%$ region                    & $(-3.390,\ 3.112)$ & $(-6.224,\ 6.124)$ & & \\
 & uncorrected moment               & $[-3.194,\ 2.916]$ & $[-6.041,\ 5.941]$ & & \\
\midrule
\multicolumn{2}{l}{observed-wage estimate} & $-0.210\ (0.008)$ & $0.033\ (0.015)$ & & \\
\bottomrule
\end{tabular}

\caption*{\footnotesize Notes. CPS 2015, $N=32{,}523$ workers, $p=259$ controls. Rows labelled $\widehat\Theta_I$ and $95\%$ region report the coordinate projections of the estimated set and of the confidence region $C_N(\widehat{c}_{\tau};\widehat{\xi})$ with $\tau=0.05$. The column area\,/\,box is the ratio of the area of $\widehat\Theta_I$ to the area of that box. The critical value $\widehat{c}_{\tau}$ is the $(1-\tau)$-quantile of $B_N=203$ disjoint block statistics of size $b=160$, balanced across the two folds. The row ``uncorrected moment'' reports the projections obtained from the non-orthogonal specification. The observed-wage estimate is infeasible and is reported for reference, with heteroskedasticity-robust standard errors in parentheses.}
\end{table}

\section{Conclusion}
\label{sec:conclusion}
This paper develops estimation and inference for partially identified models in which the moment inequalities defining the identified set depend on point-identified nuisance functions estimated in a first stage. The approach combines Neyman-orthogonal moment inequalities, cross-fitting, and the criterion-function representation of partially identified sets. Under the stated conditions, estimation of the nuisance functions is second order and the feasible contour set inherits the first-order Hausdorff rate of its infeasible counterpart. For inference, we develop a disjoint-block subsampling procedure adapted to the cross-fitted first stage. It yields asymptotically valid pointwise coverage without strengthening the first-stage small-bias requirement used for set estimation. The examples illustrate different aspects of the framework. In the wage-selection application, integrating the conditional restrictions underlying \citet{BGIM} yields unconditional moment inequalities that can accommodate a rich, machine-learned conditioning set. The interquantile-range problem shows that the same framework extends beyond affine systems. The resulting identified set need not be convex, yet it can be estimated directly as the zero set of the criterion and projected onto the interquantile range without separately inverting estimated bound functions. Finally, the interval-outcome example illustrates a different source of nonregularity. The general theory proposed herein allows continuous, discrete, or mixed treatment designs provided the population remainder associated with the positive-part kink is quadratic in the first-stage error.

The results also indicate several limits of the present analysis. The subsampling guarantees are pointwise in the data-generating process, and the number of moment inequalities is fixed in the asymptotic theory. Developing uniformly valid calibration with regularized first stages, allowing the number of thresholds or directions to increase with the sample size, and extending the theory to moment systems with more general forms of parameter-dependent non-smoothness are natural directions for further work.

\newpage
\appendix

\noindent\emph{\Cref{app:theory} and \Cref{app:proofs} are intended for online publication only.}
\section{Asymptotic theory}
\label{app:theory}
\label{sec:theory}

\subsection{Conditions on the moment function and the first stage}
\label{sec:conditions}

We take the moment function to be Neyman-orthogonal in the sense of \Cref{def:ortho}, on a set of nuisance values that contains the first-stage estimates with probability approaching one. The precise requirement is part (d) of \Cref{ass:smallbiasvalue} below. The paper posits an orthogonal moment rather than constructing one. For point-identified parameters the construction is largely automatic. Orthogonal moments for many functionals have been derived, or can be derived, following \citet{chernozhukov2016double} and \citet{LRSP}. A smooth functional of a nonparametric regression has a Riesz representer, and \citet{chernozhukov2021debiased, chernozhukov2021automatic} estimate that representer directly from the data, so the orthogonal moment is produced by the procedure rather than derived by hand for each application.

\Cref{ass:partid} ensures that the population criterion distinguishes the boundary of the identified set: once $\theta$ is bounded away from $\Theta_{I}$, the moment $\mathrm{E}\,g(Z;\theta,\xi_{0})$ is bounded away from zero by an amount proportional to $d(\theta,\Theta_{I})$. It is the population analogue of Condition~C.2 of \citet[p.~1253]{CHT}, the existence of a polynomial minorant, which bounds the sample criterion below by $\kappa[d(\theta,\Theta_{I})\wedge\delta]^{\gamma}$ outside an $a_{n}^{-1/\gamma}$-neighbourhood of $\Theta_{I}$. Its counterpart in \citet{CLR} is Condition~V, $\theta_{n}(v)-\theta_{n0}\geq(c_{n}d(v,V_{0}))^{\rho_{n}}\wedge\delta$, stated there for an index set $V_{0}$ rather than for a parameter. We depart from these conditions in two respects. First, we impose the bound on the population moment rather than on the sample criterion. Second, we fix the exponent at one on the moment, hence $\gamma=2$ on the criterion, and \Cref{sec:verification} shows that this is not a restriction for affine moment systems or for the interquantile moments of \Cref{sec:example1-iqr}. We keep the truncation at $\delta_{\min}$, as both of them do, so nothing is required of $\theta$ far from $\Theta_{I}$ beyond a fixed positive floor.

\begin{assumption}[Identifiability]
\label{ass:partid}
The identified set $\Theta_{I}$ is non-empty, and there exist constants $C_{\min}>0$ and $\delta_{\min}>0$ such that
\begin{align}
\label{eq:partid}
\big\|\mathrm{E}\,g(Z;\theta,\xi_{0})\big\|_{+} \;\geq\; C_{\min}\,\big[d(\theta,\Theta_{I})\wedge\delta_{\min}\big], \qquad \forall\,\theta\in\Theta.
\end{align}
\end{assumption}

\Cref{ass:donsker} is a functional central limit theorem for the moment function evaluated at the true nuisance.

\begin{assumption}[Donsker Property]
\label{ass:donsker}
In the space $\ell^{\infty}(\Theta)^{L}$,
\begin{align}
\label{eq:pdonsker}
\mathbb{G}_{N}\,g(Z_{i};\theta,\xi_{0}) := \sqrt{N}\big(\mathrm{E}_{N}\,g(Z_{i};\theta,\xi_{0})-\mathrm{E}\,g(Z_{i};\theta,\xi_{0})\big) \;\Rightarrow\; \mathbb{G}(\theta),
\end{align}
where $\mathbb{G}(\theta)$ is a mean-zero Gaussian process on $\Theta$ with almost surely continuous paths.
\end{assumption}

\Cref{ass:smallbiasvalue} controls the impact of first-stage estimation. It is stated on a sequence of neighbourhoods $\Xi_{N}\subset\Xi$ of $\xi_{0}$ that contain the estimate $\widehat{\xi}$ with probability approaching one and shrink as $N$ grows. The rate of shrinkage $g_{N}$ measures the quality of the first stage. The assumption combines Neyman orthogonality, which removes the first-order effect of the estimation error $\widehat{\xi}-\xi_{0}$, with a bound on the second-order remainder, condition~(e).

\begin{assumption}[Orthogonality and First-Stage Quality]
\label{ass:smallbiasvalue}
There exists a sequence of neighbourhoods $\Xi_{N}\subset\Xi$ of $\xi_{0}$ such that the following hold.
\begin{enumerate}
\item[(a)] \emph{(Truth in the neighbourhood.)} The true value $\xi_{0}$ belongs to $\Xi_{N}$ for all $N\geq 1$.
\item[(b)] \emph{(Estimator in the neighbourhood.)} There exists a sequence $\phi_{N}=o(1)$ such that every fold estimate $\widehat{\xi}^{(k)}$, $k=1,\ldots,K$, belongs to $\Xi_{N}$ with probability at least $1-\phi_{N}$. We write $\mathcal{B}_{N}:=\{\widehat{\xi}^{(k)}\in\Xi_{N}\text{ for all }k\}$, so that $\Pr(\mathcal{B}_{N})\geq 1-K\phi_{N}\to 1$.
\item[(c)] \emph{(Consistency.)} The neighbourhoods shrink:
\begin{equation*}
    g_{N} := \sup_{\xi\in\Xi_{N}}\|\xi-\xi_{0}\|_{P,2} \;=\; o(1).
\end{equation*}
\item[(d)] \emph{(Orthogonality.)} The moment function $g(Z;\theta,\xi)$ is Neyman-orthogonal at $\xi_{0}$ on $\Xi_{N}$ in the sense of \Cref{def:ortho}, for every $\theta\in\Theta$.
\item[(e)] \emph{(Small bias.)} There exists a sequence $\bar{s}_{N}=o(N^{-1/2})$ bounding the bias induced by the first stage:
\begin{equation*}
    \sup_{\theta\in\Theta}\sup_{\xi\in\Xi_{N}}\big\|\mathrm{E}\,g(Z;\theta,\xi)-\mathrm{E}\,g(Z;\theta,\xi_{0})\big\| \;\leq\; \bar{s}_{N}.
\end{equation*}
\end{enumerate}
\end{assumption}

Only conditions~(b) and~(e) enter the proofs. Condition~(e) is the only channel through which the first stage enters the bias. Conditions~(a), (c) and~(d) describe the mechanism that makes it attainable at the rate $o(N^{-1/2})$, and they are verified alongside it in \Cref{sec:verification}. When $r\mapsto\mathrm{E}\,g(Z;\theta,\xi_{0}+r(\xi-\xi_{0}))$ is twice differentiable, condition~(d) removes the first-order term and a Taylor expansion gives
\begin{equation}\label{eq:taylor}
\big\|\mathrm{E}\,g(Z;\theta,\xi)-\mathrm{E}\,g(Z;\theta,\xi_{0})\big\| \;\leq\; \tfrac{1}{2}\sup_{r\in[0,1)}\big\|\partial_{r}^{2}\,\mathrm{E}\,g(Z;\theta,\xi_{0}+r(\xi-\xi_{0}))\big\|,
\end{equation}
so a bound on the second Gateaux derivative implies condition~(e). When that derivative is a bounded bilinear form in the nuisance error, condition~(e) holds with $\bar{s}_{N}=Cg_{N}^{2}$, which is the product-of-rates condition of the double/debiased machine-learning literature. We state condition~(e) as a bias bound rather than a derivative bound because the bracketed-wage moment of \Cref{sec:example2} has a kink and is verified in \Cref{sec:verif-emp} without differentiability. No rate is imposed in condition~(c): the first-stage rate enters the variance only through \Cref{ass:concentration}(c).

\Cref{ass:concentration} requires that the moment function $g(Z;\theta,\xi)$ be regular in $\theta$ for each fixed $\xi\in\Xi_{N}\cup\{\xi_{0}\}$, with constants that do not depend on $\xi$. We consider the class $\mathcal{F}_{\xi}:=\{g(\cdot;\theta,\xi):\theta\in\Theta\}$ and require its uniform covering entropy to be bounded. Uniformity over $\xi$ lets the proofs condition on a fold's estimate and integrate it out.

\begin{assumption}[Regularity of Moment Function]
\label{ass:concentration}
The following hold for the class $\mathcal{F}_{\xi}$, for every $N$ and every $\xi\in\Xi_{N}\cup\{\xi_{0}\}$, with constants $(c,c_{1},a,v)$ that depend on neither $\xi$ nor $N$.
\begin{enumerate}
\item[(a)] \emph{(Envelope.)} There exists a measurable envelope $F_{\xi}=F_{\xi}(Z)$ that bounds all elements in the function class almost surely:
\begin{equation*}
    \sup_{\theta\in\Theta}|g_{l}(Z;\theta,\xi)| \;\leq\; F_{\xi}(Z) \quad \text{a.s.}, \qquad l\in\{1,\dots,L\},
\end{equation*}
where $g=(g_{1},\dots,g_{L})$. Moreover, the envelope $F_{\xi}$ has a finite $c$-norm for some $c>2$: $\|F_{\xi}\|_{P,c}:=(\int|F_{\xi}(z)|^{c}\,dP(z))^{1/c}\leq c_{1}$.
\item[(b)] \emph{(Entropy.)} There exist finite constants $a\geq e$ and $v\geq 1$ such that the uniform covering entropy of the class $\mathcal{F}_{\xi}$ is bounded:
\begin{align}
\label{eq:entropy}
\sup_{\tilde{Q}}\log \mathcal{N}\big(\epsilon\|F_{\xi}\|_{\tilde{Q},2},\;\mathcal{F}_{\xi},\;\|\cdot\|_{\tilde{Q},2}\big) \;\leq\; v\log(a/\epsilon), \qquad 0<\epsilon\leq 1,
\end{align}
where $\mathcal{N}(\epsilon,\mathcal{F},\|\cdot\|)$ is the $\epsilon$-covering number of $\mathcal{F}$ in the norm $\|\cdot\|$ and the supremum is over finitely discrete probability measures $\tilde{Q}$.
\item[(c)] \emph{(Continuity in the nuisance.)} There exists a sequence $r_{N}'$ obeying $r_{N}'\sqrt{\log N}=o(1)$ that bounds the mean-square distance between the moment at a nuisance in the neighbourhood and at the truth:
\begin{equation*}
    \sup_{\theta\in\Theta}\sup_{\xi\in\Xi_{N}}\big(\mathrm{E}\|g(Z;\theta,\xi)-g(Z;\theta,\xi_{0})\|^{2}\big)^{1/2} \;\leq\; r_{N}'.
\end{equation*}
\end{enumerate}
\end{assumption}

\Cref{ass:concentration} generalizes the regularity assumption in the point-identified moment problem of \citet{chernozhukov2016double} to the partially identified case.

\subsection{Rate of convergence}
\label{sec:main}

The first result is the rate at which the set estimator of \Cref{def:setestim} converges to $\Theta_I$ in Hausdorff distance. Because $g$ is Neyman-orthogonal in $\xi$, the first-stage estimation error does not affect the first-order asymptotics, and the criterion-function theory of \citet{CHT} applies to the feasible criterion. The rate is that of the infeasible contour set built on the true nuisance, the case $a_{n}=N$, $\gamma=2$ of Theorem~3.1(2) in \citet{CHT}. The level must satisfy the two hypotheses of that theorem: it contains the infeasible statistic with probability approaching one, which makes $\widehat{\Theta}_{I}$ cover $\Theta_{I}$, and it is of smaller order than $N$.

\begin{theorem}[Consistency and Rate of Convergence]
\label{thm:main:ineq}
Suppose \Crefrange{ass:partid}{ass:concentration} hold. Let $\widehat{\Theta}_{I}$ be the contour set estimator of \Cref{def:setestim} with $\widehat{c}$ satisfying
\begin{align}
\label{eq:minc:inf}
\Pr\!\Big(\sup_{\theta\in\Theta_{I}}N\,Q_{N}(\theta,\widehat{\xi})\leq\widehat{c}\Big) \;=\; 1-o(1),
\qquad \widehat{c}/N \to_{p} 0.
\end{align}
Then $\Theta_{I}\subseteq\widehat{\Theta}_{I}$ with probability approaching one, and the Hausdorff distance between $\widehat{\Theta}_{I}$ and $\Theta_{I}$ satisfies
\begin{align}
d_{H}(\widehat{\Theta}_{I},\Theta_{I}) = O_{P}\!\left(\sqrt{(1\vee\widehat{c})/N}\right).
\end{align}
\end{theorem}
Any deterministic level $c_{N}\to\infty$ with $c_{N}/N\to 0$ satisfies \eqref{eq:minc:inf}, because $\sup_{\theta\in\Theta_{I}}N\,Q_{N}(\theta,\widehat{\xi})=O_{P}(1)$ under the same assumptions (Step~4 of the proof of \Cref{lem:rate}).

\Cref{thm:main:ineq} says that the level should be as small as the coverage requirement allows. A level $\widehat{c}=O_{P}(1)$ that covers delivers the parametric rate $N^{-1/2}$, but no feasible level does both in general: the infeasible statistic $\sup_{\theta\in\Theta_{I}}N\,Q_{N}(\theta,\widehat{\xi})$ is $O_{P}(1)$ but unknown, a deterministic level $c_{N}\to\infty$ covers but is unbounded, and the subsampling level $\widehat{c}_{\tau}$ of \Cref{sec:mi-subsampling} is bounded but covers only with probability $1-\tau$. For $c_{N}=\log N$ the rate is $\sqrt{\log N/N}$, slower than the parametric rate by the factor $\sqrt{\log N}$. For the confidence region $C_{N}(\widehat{c}_{\tau};\widehat{\xi})$ the proof of \Cref{thm:main:ineq} still gives the outer bound $\sup_{\theta\in C_{N}(\widehat{c}_{\tau};\widehat{\xi})}d(\theta,\Theta_{I})=O_{P}(N^{-1/2})$, since that direction does not use containment.

The parametric rate is attainable without a covering level when the criterion is \emph{degenerate}: the sample criterion vanishes on data-dependent sets that approximate $\Theta_{I}$ at rate $N^{-1/2}$, so that any bounded level that is at least the sample minimum already captures those sets.

\begin{definition}[Degeneracy]
\label{def:degeneracy}
There exists a sequence of possibly data-dependent subsets $\Theta_{N}\subset\Theta$ such that the cross-fitted criterion is minimal on $\Theta_{N}$ with probability approaching one and $\Theta_{N}$ approximates $\Theta_{I}$ at the parametric rate: for all $p\in(0,1)$, there exists $N_{p}$ such that for all $N\geq N_{p}$,
\begin{align}
\Pr\!\left(\sup_{\theta\in\Theta_{N}}\big\{Q_{N}(\theta,\widehat{\xi})-\inf_{\tilde{\theta}\in\Theta}Q_{N}(\tilde{\theta},\widehat{\xi})\big\}=0\right) \;\geq\; 1-p,
\qquad d_{H}(\Theta_{N},\Theta_{I})=O_{P}(N^{-1/2}).
\end{align}
This is Condition~C.3 of \citet{CHT} with $a_{n}=N$ and $\gamma=2$.
\end{definition}

\begin{lemma}[Sufficient Conditions for Degeneracy]
\label{lem:degeneracy}
Suppose \Cref{ass:partid,ass:donsker,ass:smallbiasvalue,ass:concentration} hold. If in addition $\Theta_{I}$ has non-empty interior and there exist positive constants $C,M,\delta$ such that, for all $\epsilon\in[0,\delta]$, $\Theta_{I}^{-\epsilon}\neq\varnothing$ and
\begin{align}
\label{eq:deg}
\max_{l}\,\mathrm{E}\,g_{l}(Z;\theta,\xi_{0}) \;\leq\; -C\,(\epsilon\wedge\delta) \quad \text{for all }\theta\in\Theta_{I}^{-\epsilon}, \qquad d_{H}(\Theta_{I}^{-\epsilon},\Theta_{I}) \;\leq\; M\epsilon,
\end{align}
where $\Theta_{I}^{-\epsilon}:=\{\theta\in\Theta_{I}:B(\theta,\epsilon)\subseteq\Theta_{I}\}$, with $B(\theta,\epsilon)$ the closed ball in $\mathbb{R}^{d_\theta}$, is the $\epsilon$-contraction of $\Theta_I$ in the sense of \citet{CHT}, then the degeneracy condition of \Cref{def:degeneracy} holds. Moreover, for any level $\widehat{c}'=O_{P}(1)$ with
\begin{align}
\label{eq:minc:deg}
\widehat{c}' \;\geq\; N\min_{\theta\in\Theta}Q_{N}(\theta,\widehat{\xi}) \qquad \text{with probability approaching one},
\end{align}
the contour set satisfies $d_{H}(C_{N}(\widehat{c}';\widehat{\xi}),\Theta_{I})=O_{P}(N^{-1/2})$.
\end{lemma}

\citet[p.~1255]{CHT} impose Condition~C.3 rather than derive it, and motivate it by the interval-data examples, where the criterion is identically zero on the interior of the identified set with probability approaching one. Their Theorem~3.2 shows that under it the contour set at any level $\widehat{c}'\geq\inf_{\theta}a_{n}Q_{n}(\theta)$ with $\widehat{c}'=O_{P}(1)$ attains the exact rate $a_{n}^{-1/\gamma}$. \Cref{lem:degeneracy} is the two-stage version of that result: it gives primitive sufficient conditions for degeneracy when the criterion carries an estimated nuisance, and its level condition \eqref{eq:minc:deg} is the analogue of theirs. The gain is in the \emph{estimation} rate only: a level $\widehat{c}'$ satisfying \eqref{eq:minc:deg} may be far too small to satisfy \eqref{eq:minc:inf}, so $C_{N}(\widehat{c}';\widehat{\xi})$ is a set estimator and not a confidence region. The level $\widehat{c}'=0$ is admissible whenever the sample criterion attains zero, which is the interval-data case.

\subsection{Subsampling inference}
\label{sec:mi-subsampling}

To construct a confidence region $C_{N}(\widehat{c}_{\tau};\widehat{\xi})$ with the coverage property
\begin{align}
\label{eq:conf}
\Pr\!\big(\Theta_{I}\subseteq C_{N}(\widehat{c}_{\tau};\widehat{\xi})\big) \;=\; 1-\tau+o(1),
\end{align}
we must find the asymptotic distribution of the inferential statistic
\begin{align}
\mathcal{C}_{N} := \sup_{\theta\in\Theta_{I}} N\,Q_{N}(\theta,\xi_{0})
\end{align}
and use it to estimate the $(1-\tau)$-quantile. Define the limiting random variable
\begin{align}
\label{eq:infstat:moment}
\mathcal{T} := \sup_{\theta\in\Theta_{I}}\|I(\theta)+\mathbb{G}(\theta)\|_{+}^{2},
\end{align}
where $I(\theta)$ is an $L$-vector with components, and the convention $-\infty+x=-\infty$, $\max\{-\infty,0\}=0$ is used inside $\|\cdot\|_{+}$,
\begin{equation}
    I_{l}(\theta) := \begin{cases} -\infty, & \text{if } \mathrm{E}\,g_{l}(Z;\theta,\xi_{0})<0, \\ 0, & \text{if } \mathrm{E}\,g_{l}(Z;\theta,\xi_{0})=0. \end{cases}
\end{equation}

The convergence $\mathcal{C}_{N}\Rightarrow\mathcal{T}$ is Condition~C.4 of \citet{CHT}, who assume it in their general theory and verify it for moment-condition models in their Section~4, where the limit takes the form \eqref{eq:infstat:moment}. In a similar fashion, we state it as \Cref{ass:c4}, with continuity of the cdf required on $(0,\infty)$ only, so that an atom at zero is allowed. \Cref{thm:main:subs} requires in addition that the cdf be strictly increasing near the quantile being estimated, the condition under which the subsampling quantile converges.

The critical value is the subsampling quantile of \Cref{def:subsampling}, computed at a deterministic preliminary level $c_N\to\infty$ with $c_N/N\to 0$, for instance $c_N=\log N$, which delivers the rate $\sqrt{\log N/N}$ of \Cref{thm:main:ineq}. We depart from \citet{CHT} in one implementation detail: their Theorem~3.3 is stated for all $\binom{N}{b}$ subsets of size $b$, with a randomly chosen subcollection allowed in practice, whereas we use $\lfloor N/b\rfloor$ \emph{disjoint} blocks. The reason is computational: with a cross-fitted first stage, each subsample statistic requires re-evaluating the criterion, so the number of subsamples is the binding cost, and the price is a larger Monte Carlo error in $\widehat{G}$, of order $B_{N}^{-1/2}$. Disjoint blocks are also independent, which Step~2 of the proof of \Cref{thm:main:subs} uses. The critical value must be an upper quantile: coverage obtains on the event $\{\mathcal{C}_{N}\leq\widehat{c}_{\tau}\}$, so $\Pr(\mathcal{C}_{N}\leq\widehat{c}_{\tau})\to 1-\tau$ requires $\widehat{c}_{\tau}\to_{p}q_{1-\tau}(\mathcal{T})$.

Subsampling also requires that the feasible criterion stay close to the infeasible criterion, at the block scale $b$, on the $\epsilon_{N}$-expansion of the identified set, $\epsilon_{N}:=\sqrt{(1\vee c_{N})/N}$, where the block statistics are evaluated. The proof requires one condition for this beyond those of \Cref{thm:main:ineq}, a linear majorant that reverses \Cref{ass:partid}.

\begin{assumption}[Linear Majorant]
\label{ass:dominance}
There exists $C_{\max}<\infty$ such that $\|\mathrm{E}\,g(Z;\theta,\xi_{0})\|_+\leq C_{\max}\,d(\theta,\Theta_{I})$ for all $\theta\in\Theta$.
\end{assumption}

No strengthening of the first-stage conditions is required for coverage. On a block of size $b$ the first-stage bias enters at scale $\sqrt{b}\,\bar{s}_{N}$ and the expansion of the identified set at scale $\sqrt{b}\,\epsilon_{N}=\sqrt{b(1\vee c_{N})/N}$, the first vanishes because $b\leq N$, and the block-size condition $b\,c_{N}/N\to 0$ of \Cref{thm:main:subs} makes the second vanish. The full-sample analogue of this step would require $\bar{s}_{N}=o(N^{-1/2}(\log N)^{-1/2})$ at $c_{N}=\log N$, which the subsampling construction does not require (\Cref{lem:rateineq}).

Finally, we impose two conditions as Conditions~C.4 and C.5 of \citet{CHT}, stated at the block size.

\begin{assumption}[Limit Distribution]
\label{ass:c4}
    The infeasible inferential statistic $\mathcal{C}_{N}:=\sup_{\theta\in\Theta_{I}}N\,Q_{N}(\theta,\xi_{0})$ converges in distribution to the law $\mathcal{T}$ of \eqref{eq:infstat:moment}, and the cdf of $\mathcal{T}$ is continuous on $(0,\infty)$. An atom at zero is allowed. It occurs exactly when, with positive probability, every binding coordinate of $\mathbb{G}(\theta)$ is non-positive at every $\theta\in\Theta_{I}$. This is Condition~C.4 of \citet{CHT}, who verify it for moment-condition models in their Section~4, with continuity required on $(0,\infty)$ rather than on $[0,\infty)$.
\end{assumption}

\begin{assumption}[Approximability]
\label{ass:c5}
Let $\epsilon_{N}:=\sqrt{(1\vee c_{N})/N}$ be the Hausdorff rate of \Cref{thm:main:ineq} at the level $c_{N}$, and let $b=b_{N}\to\infty$ with $b\,c_{N}/N\to 0$. For every fixed $C_{0}>0$ and every sequence of non-empty compact sets $K_{N}\subset\Theta$ with $d_{H}(K_{N},\Theta_{I})\leq C_{0}\epsilon_{N}$, which may depend on the data of the block, the block statistic $T_{b}':=\sup_{\theta\in K_{N}}b\,Q_{j,b}(\theta,\xi_{0})$ computed on a block of size $b$ satisfies
\begin{align*}
\Pr(T_{b}'\leq x)-\Pr(\mathcal{T}\leq x) = o(1) \qquad\text{at every continuity point $x$ of the cdf of $\mathcal{T}$.}
\end{align*}
The assumption has two parts. The first is that the block statistic built on the \emph{true} identified set has the limit law $\mathcal{T}$. That is \Cref{ass:c4} applied at sample size $b$, and it requires only $b\to\infty$. The second is that replacing $\Theta_{I}$ by a set within $C_{0}\epsilon_{N}$ of it, including the data-dependent sets that the two-sided bounds in the proof of \Cref{thm:main:subs} produce, does not move the limit law. Here the block size and the Hausdorff rate interact: the perturbation is negligible when $\sqrt{b}\,\epsilon_{N}\to 0$, which for $\epsilon_{N}\asymp\sqrt{\log N/N}$ holds whenever $b=o(N/\log N)$. It depends on the local geometry of $\Theta_{I}$ near the binding moments; \Cref{lem:lipschitz} gives sufficient conditions.
\end{assumption}

When $\theta\mapsto g(z;\theta,\xi)$ is discontinuous, as for the interquantile moments \eqref{eq:iqrmoments}, contour sets need not be compact, and we use a uniform form of \Cref{ass:c5}. For $C_{0}>0$, let $\mathcal{S}_{N}(C_{0})$ be the collection of non-empty sets $S\subset\Theta$ with $d_{H}(S,\Theta_{I})\leq C_{0}\epsilon_{N}$, and let $T_{b}:=\sup_{\theta\in\Theta_{I}}b\,Q_{j,b}(\theta,\xi_{0})$. The condition is that $T_{b}\Rightarrow\mathcal{T}$ and that, for every fixed $C_{0}>0$, there is a function $\Delta_{b}$ of the block data, the same for every block, with $\Delta_{b}\to_{p}0$ and
\begin{equation}\label{eq:c5unif}
\sup_{S\in\mathcal{S}_{N}(C_{0})}\bigg|\sup_{\theta\in S}\sqrt{b\,Q_{j,b}(\theta,\xi_{0})}-\sqrt{T_{b}}\bigg| \;\leq\; \Delta_{b} .
\end{equation}
This condition implies \Cref{ass:c5}, since every admissible $K_{N}$ belongs to $\mathcal{S}_{N}(C_{0})$.
\begin{theorem}[Validity of Subsampling]
\label{thm:main:subs}
Suppose \Crefrange{ass:partid}{ass:concentration} and \Cref{ass:dominance,ass:c4} hold, that either $\theta\mapsto g(z;\theta,\xi)$ is continuous for every $z$ and every $\xi\in\Xi_{N}\cup\{\xi_{0}\}$ and \Cref{ass:c5} holds, or \eqref{eq:c5unif} holds, and that the block size satisfies $b\to\infty$ with $b\,c_{N}/N\to 0$. Let $\tau\in(0,1)$, with $1-\tau$ the nominal confidence coefficient, and suppose $\Pr(\mathcal{T}=0)<1-\tau$ and that the cdf of $\mathcal{T}$ is strictly increasing on a neighbourhood of $q_{1-\tau}(\mathcal{T})$. Then:
\begin{enumerate}
\item[(1)] The critical value $\widehat{c}_{\tau}$ of \Cref{def:subsampling} converges in probability to the $(1-\tau)$-quantile $q_{1-\tau}(\mathcal{T})$ of $\mathcal{T}$, where $\mathcal{T}$ is given in \eqref{eq:infstat:moment}.
\item[(2)] $\Pr(\Theta_{I}\subseteq C_{N}(\widehat{c}_{\tau};\widehat{\xi}))= 1-\tau+o(1)$.
\end{enumerate}
\end{theorem}

\Cref{thm:main:subs} is pointwise in $P$, where the coverage probability converges to $1-\tau$ for each fixed $P_{0}$ satisfying the assumptions. Subsampling can deliver uniformly valid inference in moment-inequality models, including coverage of the identified set, under conditions given by \citet{RomanoShaikh2008,RomanoShaikh2010} and \citet{AndrewsGuggenberger2009}; bootstrap and simulation alternatives with uniform coverage of the identified set are developed by \citet{Bugni2010} and \citet{ChenTamerChristensen}. Two parts of our argument are stated for fixed $P_0$ and would have to be strengthened for a uniform result. First, \Cref{ass:partid} holds with constants $C_{\min}$ and $\delta_{\min}$ that are fixed under $P_{0}$; along drifting sequences $P_{N}$ in which a moment becomes binding at rate $N^{-1/2}$ these constants degenerate, and the separation used in Step~3 of the proof of \Cref{lem:rate} fails. Second, the first-stage conditions in \Cref{ass:smallbiasvalue,ass:concentration} would have to hold uniformly over that class. What our argument shows is narrower: for fixed $P_{0}$, the bracketing in the proof of \Cref{thm:main:subs} reduces the feasible block statistics to infeasible ones computed at the true nuisance. Whether this reduction, under the two strengthenings, combines with the uniform subsampling results of \citet{RomanoShaikh2010} to give uniform coverage is a conjecture that we do not pursue.

\section{Proofs}
\label{app:proofs}
\label{appendix:moment-ineq}

This appendix proves the results of \Cref{sec:theory} for the set estimator $\widehat{\Theta}_{I}$ of \Cref{def:setestim}. \Cref{sec:auxiliary} states high-level conditions, adapted from \citet{CHT} and \citet{CLR}, and three auxiliary lemmas. \Cref{subsec:proof:aux} proves the lemmas, \Cref{sec:proofs} proves \Cref{thm:main:ineq}, \Cref{lem:degeneracy} and \Cref{thm:main:subs}, and \Cref{sec:verification} verifies the conditions for the examples.

\subsection{Auxiliary statements}
\label{sec:auxiliary}

All statements below are for the $K$-fold cross-fitted criterion
\begin{align}
\label{eq:QN:appendix}
Q_{N}(\theta,\widehat{\xi}) &:= \bigg\|\frac{1}{N}\sum_{k=1}^{K}\sum_{i\in J_{k}}g\big(Z_{i};\theta,\widehat{\xi}^{(k)}(Z_{i})\big)\bigg\|_{+}^{2}, \nonumber\\
Q_{N}(\theta,\xi) &:= \bigg\|\frac{1}{N}\sum_{i=1}^{N}g\big(Z_{i};\theta,\xi(Z_{i})\big)\bigg\|_{+}^{2}
\end{align}
for a deterministic $\xi\in\Xi$, where $\widehat{\xi}^{(k)}$ is built from $J_{k}^{c}:=\{1,\dots,N\}\setminus J_{k}$ only. The number of folds $K\geq 2$ is fixed and, for simplicity, divides $N$ and the block size $b$, so that the weights $1/K$ below equal the fold shares. Every empirical-process bound below is first established fold by fold, conditionally on $J_{k}^{c}$. Given $J_{k}^{c}$, $\widehat{\xi}^{(k)}$ is a fixed function, which lies in $\Xi_{N}$ on the event $\mathcal{B}_{N}$, and the $n$ observations of $J_{k}$ are i.i.d.\ and independent of it. A bound that holds with conditional probability at least $1-p$ on the $J_{k}^{c}$-measurable event $\mathcal{B}_{N}^{(k)}:=\{\widehat{\xi}^{(k)}\in\Xi_{N}\}$ holds unconditionally with probability at least $(1-p)(1-\phi_{N})$, since $\Pr(A)\geq\mathrm{E}[\mathbf{1}\{\mathcal{B}_{N}^{(k)}\}\Pr(A\mid J_{k}^{c})]$, and the $K$ fold-level bounds are combined by the triangle inequality at the cost of a factor $K$ absorbed into the constants. To keep the notation light we suppress the fold index and write $\widehat{\xi}_{i}$ for $\widehat{\xi}^{(k(i))}(Z_{i})$. We write $A_{N}\lesssim_{P}B_{N}$ when $A_{N}=O_{P}(B_{N})$.

Expectations of functions of the cross-fitted nuisance are fold-conditional throughout: for a function $h$,
\begin{align}
\label{eq:foldE}
\mathrm{E}\,h(Z_{i},\theta,\widehat{\xi}_{i}) \;&:=\; \frac{1}{K}\sum_{k=1}^{K}\mathrm{E}\big[h(Z,\theta,\widehat{\xi}^{(k)})\,\big|\,J_{k}^{c}\big], \\
\mathbb{G}_{N}\,h(Z_{i},\theta,\widehat{\xi}_{i}) \;&:=\; \sqrt{N}\big(\mathrm{E}_{N}-\mathrm{E}\big)h(Z_{i},\theta,\widehat{\xi}_{i}),
\end{align}
so that $N\,Q_{N}(\theta,\widehat{\xi})=\|\mathbb{G}_{N}\,g(Z_{i};\theta,\widehat{\xi}_{i})+\sqrt{N}\,\mathrm{E}\,g(Z_{i};\theta,\widehat{\xi}_{i})\|_{+}^{2}$ holds by definition, and on $\mathcal{B}_{N}$ the fold-conditional bias obeys $\|\mathrm{E}\,g(Z_{i};\theta,\widehat{\xi}_{i})-\mathrm{E}\,g(Z_{i};\theta,\xi_{0})\|\leq\bar{s}_{N}$ by \Cref{ass:smallbiasvalue}(e), since each of the $K$ conditional biases does. The high-level conditions below are stated for the cross-fitted criterion $Q_{N}(\theta,\widehat{\xi})$ itself, not for a deterministic nuisance: with $K\geq 2$ folds the sample is not i.i.d.\ conditionally on $\widehat{\xi}$, so a statement for deterministic $\xi\in\Xi_{N}$ does not transfer to $\widehat{\xi}$ by conditioning, and the argument is carried out fold by fold.

\Cref{ass:cons,ass:noeffect,ass:fast} below are high-level conditions on the population and sample criteria. \Cref{lem:rate} derives the rate from them, \Cref{lem:rateineq} verifies \Cref{ass:noeffect,ass:fast} from the conditions of \Cref{sec:theory}, and the proof of \Cref{thm:main:ineq} verifies \Cref{ass:cons}.

Rates are stated in the Hausdorff distance \eqref{eq:hausdorff}. As in \Cref{sec:framework}, all random quantities below are assumed measurable, including suprema and infima over subsets of $\Theta$.

\begin{assumption}[Consistency of the Criterion]
\label{ass:cons}
We impose the following:
\begin{enumerate}
\item[(a)] \emph{(Regularity of $Q(\theta,\xi_{0})$.)} $Q(\theta,\xi_{0})$ is non-negative, and for every $\epsilon>0$ there exists $\varpi(\epsilon)>0$ such that
\begin{align}
\label{eq:cons1}
\inf_{\theta\in\Theta\setminus\Theta_{I}^{\epsilon}}Q(\theta,\xi_{0}) \;\geq\; \varpi(\epsilon) > 0.
\end{align}
\item[(b)] \emph{(Fast convergence on $\Theta_{I}$.)} Uniformly on $\Theta_{I}$,
\begin{align}
\label{eq:cons2}
\sup_{\theta\in\Theta_{I}}N\,Q_{N}(\theta,\xi_{0}) = O_{P}(1).
\end{align}
\item[(c)] \emph{(Slow convergence on $\Theta$.)} Uniformly on $\Theta$,
\begin{align}
\label{eq:cons3}
\sup_{\theta\in\Theta}\big(Q(\theta,\xi_{0})-Q_{N}(\theta,\xi_{0})\big)_{+} = O_{P}(N^{-1/2}).
\end{align}
\item[(d)] \emph{(Separation away from $\Theta_{I}$.)} There exist $\omega,\kappa>0$ such that for any $p\in(0,1)$ there exist $D_{p},N_{p}$ with, for all $N\geq N_{p}$,
\begin{align}
\label{eq:cons4}
\Pr\!\Big(Q_{N}(\theta,\widehat{\xi}) \;\geq\; \kappa\,[d(\theta,\Theta_{I})\wedge\omega]^{2}\quad\forall\,\theta:\;d(\theta,\Theta_{I})\geq(D_{p}/N)^{1/2}\Big) \;\geq\; 1-p.
\end{align}
\end{enumerate}
\end{assumption}

\begin{assumption}[No Effect of First-Stage Estimation Error]
\label{ass:noeffect}
We impose the following:
\begin{enumerate}
\item[(a)] \emph{(Slow convergence on $\Theta$.)} For any $p>0$ there exist $r_{p},N_{p}$ such that for all $N\geq N_{p}$,
\begin{align*}
\Pr\!\Big(\sqrt{N}\,\sup_{\theta\in\Theta}\big|Q_{N}(\theta,\widehat{\xi})-Q_{N}(\theta,\xi_{0})\big| \;\leq\; r_{p}\Big) \;\geq\; 1-p.
\end{align*}
\item[(b)] \emph{(Fast convergence on $\Theta_{I}$.)} For any $p,\epsilon>0$ there exists $N_{p,\epsilon}$ such that for all $N\geq N_{p,\epsilon}$,
\begin{align*}
\Pr\!\Big(N\,\sup_{\theta\in\Theta_{I}}\big|Q_{N}(\theta,\widehat{\xi})-Q_{N}(\theta,\xi_{0})\big| \;\leq\; \epsilon\Big) \;\geq\; 1-p.
\end{align*}
\end{enumerate}
\end{assumption}

\begin{assumption}[Fast Convergence at Block Size $b$]
\label{ass:fast}
Let $\epsilon_{N}:=\sqrt{(1\vee c_{N})/N}$ be the Hausdorff rate delivered by \Cref{thm:main:ineq} at the preliminary level $c_{N}$ of \Cref{def:subsampling}, up to a constant. For the recommended level $c_{N}=\log N$ this is $\epsilon_{N}=\sqrt{\log N/N}$. Let $b=b_{N}\to\infty$ with $b\,c_{N}/N\to 0$, and let $Q_{j,b}(\theta,\widehat{\xi})$ be the criterion computed on a block $j$ of size $b$, balanced across folds, with the cross-fitted nuisance of the full sample rather than one re-estimated on the block. For every fixed $C_{0}>0$ and any $p,\epsilon>0$ there exists $N_{p,\epsilon,C_{0}}$ such that for all $N\geq N_{p,\epsilon,C_{0}}$ and every block $j$,
\begin{align*}
\Pr\!\Big(b\sup_{\theta\in\Theta_{I}^{C_{0}\epsilon_{N}}}\big|Q_{j,b}(\theta,\widehat{\xi})-Q_{j,b}(\theta,\xi_{0})\big| \;\leq\; \epsilon\Big) \;\geq\; 1-p.
\end{align*}
\end{assumption}

The next three lemmas provide the key technical tools. The first establishes coverage, consistency, and the rate of convergence under the high-level assumptions.

\begin{lemma}[Coverage, Consistency, and Rate of Convergence]
\label{lem:rate}
Let $\widehat{c}/N\to_{p}0$ and
\begin{equation*}
\Pr\Big(\sup_{\theta\in\Theta_{I}}N\,Q_{N}(\theta,\widehat{\xi})\leq\widehat{c}\Big)=1-o(1).
\end{equation*}
Then:
\begin{enumerate}
\item[(i)] \Cref{ass:cons}(a),(c) and \Cref{ass:noeffect}(a) imply $\Theta_{I}\subseteq\widehat{\Theta}_{I}$ with probability approaching one and $d_{H}(\widehat{\Theta}_{I},\Theta_{I})=o_{P}(1)$.
\item[(ii)] \Cref{ass:cons}(d) implies $d_{H}(\widehat{\Theta}_{I},\Theta_{I})=O_{P}(\sqrt{(1\vee\widehat{c})/N})$.
\item[(iii)] If \Cref{ass:cons}(b) and \Cref{ass:noeffect}(b) hold, one can choose $\widehat{c}=O_{P}(1)$.
\end{enumerate}
\end{lemma}

The second lemma controls the concentration of the empirical process of the estimated moments around their population counterparts.

\begin{lemma}[Concentration of Estimated Moments]
\label{lem:concentrate}
Suppose \Cref{ass:smallbiasvalue,ass:concentration} hold. Let a block be a deterministic set of $m=m_{N}\to\infty$ indices with $m/K$ from every fold, and let $\mathbb{G}_{m}$ be the empirical process of the block with the full-sample cross-fitted nuisance $\widehat{\xi}_{i}$; $m=N$ gives the full sample. For every $p\in(0,1)$ there exist $C_{p}<\infty$ and $N_{p}$, depending on $(K,a,v,c,c_{1},p)$ only and not on the block, such that for all $N\geq N_{p}$, with $\zeta_{m}:=r_{N}'\sqrt{\log m}+m^{-1/2+1/c}\log m$,
\begin{align}
\label{eq:moment2}
\Pr\Big(\sup_{\theta\in\Theta}\Big\|\mathbb{G}_{m}\left\{g(Z_{i};\theta,\widehat{\xi}_{i})-g(Z_{i};\theta,\xi_{0})\right\}\Big\| \;\leq\; C_{p}\,\zeta_{m}\Big)\;\geq\;1-p .
\end{align}
\end{lemma}

The third lemma shows that the high-level \Cref{ass:noeffect,ass:fast} follow from the conditions stated in \Cref{sec:theory}.

\begin{lemma}[Sufficient Conditions for General Moment Problems]
\label{lem:rateineq}
\Cref{ass:smallbiasvalue,ass:concentration} and \eqref{eq:pdonsker} imply \Cref{ass:noeffect}, and together with \Cref{ass:dominance} and a block size $b\to\infty$ with $b\,c_{N}/N\to 0$ they imply \Cref{ass:fast}.
\end{lemma}

\subsection{Proofs of auxiliary lemmas}
\label{subsec:proof:aux}

\begin{proof}[Proof of \Cref{lem:rate}]
The proof uses two elementary inequalities: for real functions $f$ and $h$ on a set $A$,
\begin{align}
\label{eq:basic}
\sup_{A}f \;\leq\; \sup_{A}h + \sup_{A}(f-h)_{+}, \qquad
\inf_{A}f \;\geq\; \inf_{A}h - \sup_{A}(h-f)_{+}.
\end{align}
\medskip
\noindent\emph{Step 1: containment.}
The event $\{\sup_{\theta\in\Theta_{I}}N\,Q_{N}(\theta,\widehat{\xi})\leq\widehat{c}\}$ is the event $\{\Theta_{I}\subseteq\widehat{\Theta}_{I}\}$, so the level condition of the lemma is $\Pr(\Theta_{I}\subseteq\widehat{\Theta}_{I})\to 1$.
\medskip
\noindent\emph{Step 2: convergence without rate.}
Fix $\epsilon>0$. We show $\Pr(d_{H}(\widehat{\Theta}_{I},\Theta_{I})\leq\epsilon)\to 1$. By \Cref{ass:cons}(a), there exists $\varpi(\epsilon)>0$ such that
\begin{align*}
\inf_{\theta\in\Theta\setminus\Theta_{I}^{\epsilon}}Q(\theta,\xi_{0}) \;\geq\; \varpi(\epsilon) > 0.
\end{align*}
We chain inequalities using \eqref{eq:basic}:
\begin{align*}
\sup_{\theta\in\widehat{\Theta}_{I}}Q(\theta,\xi_{0})
&\leq \sup_{\theta\in\widehat{\Theta}_{I}}Q_{N}(\theta,\xi_{0}) + \sup_{\theta\in\widehat{\Theta}_{I}}\big(Q(\theta,\xi_{0})-Q_{N}(\theta,\xi_{0})\big)_{+} \tag{by \eqref{eq:basic}} \\
&\leq \sup_{\theta\in\widehat{\Theta}_{I}}Q_{N}(\theta,\xi_{0}) + O_{P}(1/\sqrt{N}) \tag{by \eqref{eq:cons3}} \\
&\leq \sup_{\theta\in\widehat{\Theta}_{I}}Q_{N}(\theta,\widehat{\xi}) + \sup_{\theta\in\widehat{\Theta}_{I}}\big|Q_{N}(\theta,\xi_{0})-Q_{N}(\theta,\widehat{\xi})\big| + O_{P}(1/\sqrt{N}) \tag{by \eqref{eq:basic}} \\
&\leq \widehat{c}/N + O_{P}(N^{-1/2}) + O_{P}(N^{-1/2}) \;=\; o_{P}(1). \tag{\Cref{ass:noeffect}(a) and $\widehat{c}/N\to_{p}0$}
\end{align*}
Hence, by \eqref{eq:cons1}, $\Pr(\widehat{\Theta}_{I}\subseteq\Theta_{I}^{\epsilon})=1-o(1)$.
\medskip
\noindent\emph{Step 3: convergence rate.}
Fix $p\in(0,1)$, let $\kappa,\omega$ be as in \Cref{ass:cons}(d), let $D_{p/3}$ be as in \eqref{eq:cons4}, and set $\epsilon_{N,p}:=\big(2\max\{D_{p/3},\,\widehat{c}/\kappa\}/N\big)^{1/2}$. We show that $\Pr(d_{H}(\widehat{\Theta}_{I},\Theta_{I})\leq\epsilon_{N,p})\geq 1-p$ for all large $N$.

Since $\widehat{c}/N\to_{p}0$, there exists $N_{p/3}^{A}$ such that for all $N\geq N_{p/3}^{A}$,
\begin{align*}
\Pr(\mathcal{A}_{N}) := \Pr\big(\widehat{c}/(N\kappa)<\omega^{2}/8\big) \;\geq\; 1-p/3.
\end{align*}
On $\mathcal{A}_{N}$ and for $N$ large,
\begin{align*}
\epsilon_{N,p}^{2} = \frac{2\max\{D_{p/3},\;\widehat{c}/\kappa\}}{N} \;\leq\; \max\!\Big\{\frac{2D_{p/3}}{N},\;\frac{\omega^{2}}{4}\Big\} \;\leq\; \frac{\omega^{2}}{4}, \qquad \epsilon_{N,p} \;\geq\; \sqrt{D_{p/3}/N},
\end{align*}
so \Cref{ass:cons}(d) applies at radius $\epsilon_{N,p}$.

Define
\begin{align*}
\mathcal{D}_{\widehat{\xi}} := \Big\{\inf_{\theta:\;d(\theta,\Theta_{I})\geq\epsilon_{N,p}}N\,Q_{N}(\theta,\widehat{\xi}) \;\geq\; N\kappa\,(\epsilon_{N,p}\wedge\omega)^{2}\Big\}.
\end{align*}
Since $\epsilon_{N,p}\geq(D_{p/3}/N)^{1/2}$ on $\mathcal{A}_{N}$, \Cref{ass:cons}(d) gives $\Pr(\mathcal{D}_{\widehat{\xi}}\cap\mathcal{A}_{N})\geq 1-2p/3$ for $N\geq N_{p/3}$. Therefore, with probability at least $1-2p/3$ and for all $N\geq N_{p/3}\vee N_{p/3}^{A}$,
\begin{align}
\label{eq:separation}
\inf_{\theta:\;d(\theta,\Theta_{I})\geq\epsilon_{N,p}}N\,Q_{N}(\theta,\widehat{\xi}) \;\geq\; N\kappa\,(\epsilon_{N,p}\wedge\omega)^{2} \;=\; N\kappa\,\epsilon_{N,p}^{2} \;=\; 2\max\{\kappa D_{p/3},\,\widehat{c}\} \;>\; \widehat{c},
\end{align}
since $D_{p/3}>0$, which is without loss of generality in \Cref{ass:cons}(d). Every $\theta$ with $d(\theta,\Theta_{I})\geq\epsilon_{N,p}$ is therefore excluded from $\widehat{\Theta}_{I}=\{N\,Q_{N}(\theta,\widehat{\xi})\leq\widehat{c}\}$, so $\Pr(\widehat{\Theta}_{I}\subseteq\Theta_{I}^{\epsilon_{N,p}})\geq 1-2p/3$ for large $N$. Together with Step~1, this gives $d_{H}(\widehat{\Theta}_{I},\Theta_{I})\leq\epsilon_{N,p}$ with probability at least $1-p$ for large $N$, and $\epsilon_{N,p}\lesssim\sqrt{(1\vee\widehat{c})/N}$.
\medskip
\noindent\emph{Step 4: choice $\widehat{c}=O_{P}(1)$.}
That some $\widehat{c}=O_{P}(1)$ satisfies the level condition follows from
\begin{align*}
\sup_{\theta\in\Theta_{I}}N\,Q_{N}(\theta,\widehat{\xi})
&\leq \sup_{\theta\in\Theta_{I}}N\,Q_{N}(\theta,\xi_{0}) + \sup_{\theta\in\Theta_{I}}N\,\big|Q_{N}(\theta,\widehat{\xi})-Q_{N}(\theta,\xi_{0})\big| \\
&\leq O_{P}(1) + \sup_{\theta\in\Theta_{I}}N\,\big|Q_{N}(\theta,\widehat{\xi})-Q_{N}(\theta,\xi_{0})\big| \tag{\Cref{ass:cons}(b)} \\
&\leq O_{P}(1) + o_{P}(1), \tag{\Cref{ass:noeffect}(b)}
\end{align*}
so any $\widehat{c}$ that exceeds the $O_{P}(1)$ infeasible statistic with probability approaching one satisfies it, for instance $\widehat{c}:=\sup_{\theta\in\Theta_{I}}N\,Q_{N}(\theta,\widehat{\xi})$ itself, which is infeasible, or a slowly diverging deterministic sequence, which is feasible but not $O_{P}(1)$. The subsampling level $\widehat{c}_{\tau}$ of \Cref{def:subsampling} covers only with probability $1-\tau+o(1)$ and does not satisfy the level condition.\end{proof}

\begin{proof}[Proof of \Cref{lem:concentrate}]
Fix a fold $k$ and condition on the observations outside $J_{k}$, so that $\widehat{\xi}^{(k)}$ is a fixed function, which lies in $\Xi_{N}$ on $\mathcal{B}_{N}$. The observations of the block that belong to $J_{k}$ are $m/K$ i.i.d.\ draws independent of $\widehat{\xi}^{(k)}$, and for the full sample $m=N$. We bound the empirical process of these $m/K$ observations and combine the $K$ bounds at the end. Write $\widehat{\xi}$ for $\widehat{\xi}^{(k)}$. Consider the function class
\begin{align*}
\mathcal{F}_{2} \;=\; \left\{g_{l}(\cdot;\theta,\widehat{\xi})-g_{l}(\cdot;\theta,\xi_{0}) \;:\; l=1,\dots,L,\;\theta\in\Theta\right\} \;\subset\; \mathcal{F}_{\widehat{\xi}}-\mathcal{F}_{\xi_{0}}.
\end{align*}
Let $F_{2}:=F_{\widehat{\xi}}+F_{\xi_{0}}$ be an envelope for $\mathcal{F}_{2}$. This envelope satisfies the requirements of Lemma~6.2 of \citet{chernozhukov2016double}, since $\|F_{2}\|_{P,c}\leq\|F_{\widehat{\xi}}\|_{P,c}+\|F_{\xi_{0}}\|_{P,c}\leq 2c_{1}$ and
\begin{align*}
\log\sup_{\tilde{Q}}\mathcal{N}\left(\epsilon\|F_{2}\|_{\tilde{Q},2},\;\mathcal{F}_{2},\;\|\cdot\|_{\tilde{Q},2}\right) \;\leq\; 2v\log(2aL/\epsilon),
\end{align*}
where the factor $L$ accounts for the $L$ coordinates and is absorbed into the constants. On the event $\mathcal{B}_{N}$, the variance of the difference class is controlled as
\begin{align*}
\sup_{\theta\in\Theta}\mathrm{E}\big[\{g_{l}(Z;\theta,\widehat{\xi})-g_{l}(Z;\theta,\xi_{0})\}^{2}\,\big|\,J_{k}^{c}\big]
&\leq \sup_{\theta\in\Theta,\;\xi\in\Xi_{N}}\mathrm{E}\big\|g(Z;\theta,\xi)-g(Z;\theta,\xi_{0})\big\|^{2} \\
&\leq (r_{N}')^{2}.
\end{align*}
Let $m'=m/K\to\infty$ be the number of block observations in the fold and take the variance bound $\sigma_{m}:=r_{N}'\vee (m')^{-1/2+1/c}$. Lemma~6.2 of \citet{chernozhukov2016double} requires $\sup_{f\in\mathcal{F}_{2}}\mathrm{E}f^{2}\leq\sigma_{m}^{2}\leq\|F_{2}\|_{P,2}^{2}$. The first inequality is the variance bound above. For the second, we may take $F_{\xi}\geq 1$, since enlarging an envelope does not affect \Cref{ass:concentration}(a) or~(b); then $\|F_{2}\|_{P,2}\geq 1$, while $\sigma_{m}\to 0$. We apply Lemma~6.2 conditionally on $J_{k}^{c}$, with the class $\mathcal{F}_{2}$, envelope $F_{2}$, variance bound $\sigma_{m}^{2}$ and the constants $a\geq e$, $v\geq 1$ of \Cref{ass:concentration}(b), and use its tail bound with $t=p^{-2/c}$. With conditional probability at least $1-p$,
\begin{align*}
\sup_{f\in\mathcal{F}_{2}}|\mathbb{G}_{m'}f| \;\leq\; K_{p}\Big(&\sigma_{m}\sqrt{v\log(a\|F_{2}\|_{P,2}/\sigma_{m})} \\
&+v\,(m')^{-1/2+1/c}\|F_{2}\|_{P,c}\log(a\|F_{2}\|_{P,2}/\sigma_{m})\Big).
\end{align*}
Since $\sigma_{m}\geq(m')^{-1/2+1/c}$, the logarithms are at most a constant times $\log m$, so the right-hand side is $\lesssim\zeta_{m}$ with constants depending only on $(K,a,v,c,c_{1},p)$ and not on the block or the fold. The additional terms of the tail bound, $\sigma_{m}\sqrt{t}$, $(m')^{-1/2}\|\max_{i}F_{2}(Z_{i})\|_{P,c}\sqrt{t}$ and $(m')^{-1/2}\|\max_{i}F_{2}(Z_{i})\|_{P,2}\,t$, are of the same order, since $\|\max_{i\leq m'}F_{2}(Z_{i})\|_{P,c}\leq (m')^{1/c}\|F_{2}\|_{P,c}$. The bound is integrated over the conditioning at the cost of the factor $1-\phi_{N}$ as explained at the start of this appendix, and the $K$ fold bounds are combined by the triangle inequality, since $\mathbb{G}_{m}$ is the average of the $K$ fold processes scaled by $\sqrt{K}$. This gives
\begin{align}
\label{eq:moment}
\sup_{\theta\in\Theta}\Big\|\mathbb{G}_{m}\left\{g(Z_{i};\theta,\widehat{\xi}_{i})-g(Z_{i};\theta,\xi_{0})\right\}\Big\|
\;\lesssim_{P}\; \zeta_{m},
\end{align}
which is $o_{P}(1)$ for $m=N$ by \Cref{ass:concentration}(c), and for every $m\leq N$ with $m\to\infty$ since $r_{N}'\sqrt{\log m}\leq r_{N}'\sqrt{\log N}$.
\end{proof}

\begin{proof}[Proof of \Cref{lem:rateineq}]
The proof bounds the difference between the feasible and the infeasible criterion, first on the full sample, which gives \Cref{ass:noeffect}, and then on a block of size $b$, which gives \Cref{ass:fast}. Fix $p\in(0,1)$, let $\mathcal{B}_{N}$ be the event of \Cref{ass:smallbiasvalue}(b), and let $\mathcal{E}_{1,N}$ and $\mathcal{E}_{2,N}$ be the events, provided with a common constant $C_{p/4}$ by \Cref{lem:concentrate} and \eqref{eq:pdonsker}, on which
\begin{align*}
\Pr(\mathcal{E}_{1,N}) &:= \Pr\left(\sup_{\theta\in\Theta}\Big\|\mathbb{G}_{N}\left\{g(Z_{i};\theta,\widehat{\xi}_{i})-g(Z_{i};\theta,\xi_{0})\right\}\Big\| \;\leq\; C_{p/4}\,\zeta_{N}\right) \;\geq\; 1-p/4, \\
\Pr(\mathcal{E}_{2,N}) &:= \Pr\left(\sup_{\theta\in\Theta}\big\|\mathbb{G}_{N}\,g(Z_{i};\theta,\xi_{0})\big\| \;\leq\; C_{p/4}\right) \;\geq\; 1-p/4.
\end{align*}
Here $\zeta_{N}$ is $\zeta_{m}$ of \Cref{lem:concentrate} at $m=N$, so $\mathcal{E}_{1,N}$ has the stated probability by that lemma, and $\zeta_{N}=o(1)$ by \Cref{ass:concentration}(c). On $\mathcal{B}_{N}\cap\mathcal{E}_{1,N}\cap\mathcal{E}_{2,N}$, we decompose the criterion function difference. Setting $x:=\mathrm{E}_{N}\,g(Z_{i};\theta,\widehat{\xi}_{i})$ and $y:=\mathrm{E}_{N}\,g(Z_{i};\theta,\xi_{0})$,
\begin{align*}
\big|Q_{N}(\theta,\widehat{\xi})-Q_{N}(\theta,\xi_{0})\big|
&= \big|\|x\|_{+}^{2}-\|y\|_{+}^{2}\big|
\;\leq\; \|x-y\|^{2} + 2\|x-y\|\,\|y\|_{+},
\end{align*}
where the inequality uses $|\|a\|_{+}^{2}-\|b\|_{+}^{2}|\leq\|a-b\|^{2}+2\|a-b\|\,\|b\|_{+}$, a consequence of the triangle inequality for $\|\cdot\|_{+}$. The two factors are bounded as follows:
\begin{align*}
\|x-y\|
&= \Big\|\mathrm{E}\big[g(Z_{i};\theta,\widehat{\xi}_{i})-g(Z_{i};\theta,\xi_{0})\big] + \mathbb{G}_{N}\big[g(Z_{i};\theta,\widehat{\xi}_{i})-g(Z_{i};\theta,\xi_{0})\big]/\sqrt{N}\Big\| \\
&\leq \bar{s}_{N} + C_{p/4}N^{-1/2}\zeta_{N}, \\
\|y\|_{+} &= \big\|\mathbb{G}_{N}\,g(Z_{i};\theta,\xi_{0})/\sqrt{N} + \mathrm{E}\,g(Z_{i};\theta,\xi_{0})\big\|_{+}.
\end{align*}

\noindent\emph{Verification of \Cref{ass:noeffect}(a).}
Taking
\begin{align*}
r_{p} = (2C_{p/4})^{2} + 4C_{p/4}\!\left(C_{p/4}+\sup_{\theta\in\Theta}\|\mathrm{E}\,g(Z_{i};\theta,\xi_{0})\|\right),
\end{align*}
and using $\sqrt{N}\bar{s}_{N}+C_{p/4}\zeta_{N}\leq 2C_{p/4}$ for $N$ large, we obtain, on $\mathcal{B}_{N}\cap\mathcal{E}_{1,N}\cap\mathcal{E}_{2,N}$ and hence with probability at least $1-p$,
\begin{align*}
\sqrt{N}\,\sup_{\theta\in\Theta}\big|Q_{N}(\theta,\widehat{\xi})-Q_{N}(\theta,\xi_{0})\big| \;\leq\; r_{p}.\end{align*}

\noindent\emph{Verification of \Cref{ass:noeffect}(b).}
On $\Theta_{I}$, $\|\mathrm{E}\,g(Z_{i};\theta,\xi_{0})\|_{+}=0$ by definition of the identified set, so
\begin{align*}
\sup_{\theta\in\Theta_{I}}\sqrt{N}\,\|y\|_{+}
\;\leq\; \sup_{\theta\in\Theta_{I}}\|\mathbb{G}_{N}\,g(Z_{i};\theta,\xi_{0})\| \;\leq\; C_{p/4},
\end{align*}
and hence, on $\mathcal{B}_{N}\cap\mathcal{E}_{1,N}\cap\mathcal{E}_{2,N}$,
\begin{align*}
N\,\sup_{\theta\in\Theta_{I}}\big|Q_{N}(\theta,\widehat{\xi})-Q_{N}(\theta,\xi_{0})\big| \;\leq\; \big(\sqrt{N}\bar{s}_{N}+C_{p/4}\zeta_{N}\big)^{2}+2\big(\sqrt{N}\bar{s}_{N}+C_{p/4}\zeta_{N}\big)C_{p/4} \;\leq\; \epsilon\end{align*}
for $N$ large, since $\sqrt{N}\bar{s}_{N}\to 0$ and $\zeta_{N}\to 0$. This gives \Cref{ass:noeffect}(b).

\noindent\emph{Verification of \Cref{ass:fast}.}
Fix a block $j$ of size $b$, balanced across folds, and let $x$ and $y$ now denote the block averages $\mathrm{E}_{j,b}\,g(Z_{i};\theta,\widehat{\xi}_{i})$ and $\mathrm{E}_{j,b}\,g(Z_{i};\theta,\xi_{0})$. Every block draws $b/K$ observations from each fold, and those in fold $k$ are i.i.d.\ and independent of $\widehat{\xi}^{(k)}$, so \Cref{lem:concentrate} applies at $m=b$ with the full-sample nuisance: with probability at least $1-p/4$,
\begin{align*}
\sup_{\theta\in\Theta}\sqrt{b}\,\|x-y\| \;\leq\; \sqrt{b}\,\bar{s}_{N} + C_{p/4}\,\zeta_{b},
\end{align*}
with $\zeta_{b}$ the $\zeta_{m}$ of \Cref{lem:concentrate} at $m=b$ and a constant that depends on the block only through $b$. Both terms are $o(1)$: $\sqrt{b}\,\bar{s}_{N}\leq\sqrt{N}\,\bar{s}_{N}$ by \Cref{ass:smallbiasvalue}(e), $r_{N}'\sqrt{\log b}\leq r_{N}'\sqrt{\log N}$ by \Cref{ass:concentration}(c), and $b^{-1/2+1/c}\log b\to 0$ as $b\to\infty$. For $\theta\in\Theta_{I}^{C_{0}\epsilon_{N}}$,
\begin{align}
\label{eq:fastdecomp}
b\,\big|Q_{j,b}(\theta,\widehat{\xi})-Q_{j,b}(\theta,\xi_{0})\big|
&\leq b\,\|x-y\|^{2} + 2b\,\|x-y\|\,\|y\|_{+},
\end{align}
so the first term is $o(1)$. For the second factor, \eqref{eq:pdonsker} applied to the $b$ i.i.d.\ observations of the block gives an event of probability at least $1-p/4$ on which $\sup_{\theta\in\Theta}\|\mathbb{G}_{b}\,g(Z_{i};\theta,\xi_{0})\|\leq C_{p/4}$, and there \Cref{ass:dominance} gives
\begin{align*}
\sup_{\theta\in\Theta_{I}^{C_{0}\epsilon_{N}}}\sqrt{b}\,\|y\|_{+}
&\leq C_{p/4} + \sup_{\theta\in\Theta_{I}^{C_{0}\epsilon_{N}}}\sqrt{b}\,\|\mathrm{E}\,g(Z_{i};\theta,\xi_{0})\|_{+} \\
&\leq C_{p/4} + C_{\max}\,C_{0}\,\sqrt{b}\,\epsilon_{N} \;=\; C_{p/4}+o(1),
\end{align*}
because $\sqrt{b}\,\epsilon_{N}=\sqrt{b(1\vee c_{N})/N}\to 0$ by hypothesis. The second term of \eqref{eq:fastdecomp} is therefore the product of an $o(1)$ and an $O(1)$ factor. Thus, for any $\epsilon,p>0$ and $N$ large,
\begin{align*}
\Pr\left(\sup_{\theta\in\Theta_{I}^{C_{0}\epsilon_{N}}}b\,\big|Q_{j,b}(\theta,\widehat{\xi})-Q_{j,b}(\theta,\xi_{0})\big| \;\leq\; \epsilon\right) \;\geq\; 1-p,
\end{align*}
uniformly in $j$. On a block, the bias enters at scale $\sqrt{b}\,\bar{s}_{N}$ and the expansion at scale $\sqrt{b}\,\epsilon_{N}$, both smaller than their full-sample counterparts, so no condition beyond \Cref{ass:smallbiasvalue}(e), \Cref{ass:concentration}(c) and $b\,c_{N}/N\to 0$ is required. The same decomposition at $b=N$ shows that the full-sample analogue would require $\sqrt{N}\epsilon_{N}(\sqrt{N}\bar{s}_{N}+\zeta_{N})=o(1)$, a strengthening of the bias and continuity conditions by the factor $\sqrt{1\vee c_{N}}$; the subsampling construction does not require it.
\end{proof}

\subsection{Proofs of the main results}
\label{sec:proofs}

\begin{proof}[Proof of \Cref{thm:main:ineq}]
We verify that \Cref{ass:cons}(a)--(d) hold under the conditions of \Cref{sec:theory}, then apply \Cref{lem:rate}.

\medskip
\noindent\emph{Verification of \Cref{ass:cons}(a).}
By \Cref{ass:partid}, for every $\epsilon>0$,
\begin{align*}
\inf_{\theta\in\Theta\setminus\Theta_{I}^{\epsilon}}Q(\theta,\xi_{0})
\;=\; \inf_{\theta\in\Theta\setminus\Theta_{I}^{\epsilon}}\big\|\mathrm{E}\,g(Z_{i};\theta,\xi_{0})\big\|_{+}^{2}
\;\geq\; C_{\min}^{2}\,(\epsilon\wedge\delta_{\min})^{2} > 0.\end{align*}

\medskip
\noindent\emph{Verification of \Cref{ass:cons}(b).}
For $\theta\in\Theta_{I}$ we have $\mathrm{E}\,g_{l}(Z_{i};\theta,\xi_{0})\leq 0$ for every $l$, and the map $v\mapsto\|v\|_{+}$ is coordinatewise non-decreasing. Hence, writing $\sqrt{N}\,\mathrm{E}_{N}\,g=\mathbb{G}_{N}\,g+\sqrt{N}\,\mathrm{E}\,g$ and dropping the non-positive second term,
\begin{align*}
\sup_{\theta\in\Theta_{I}}N\,Q_{N}(\theta,\xi_{0})
&= \sup_{\theta\in\Theta_{I}}\big\|\mathbb{G}_{N}\,g(Z_{i};\theta,\xi_{0})+\sqrt{N}\,\mathrm{E}\,g(Z_{i};\theta,\xi_{0})\big\|_{+}^{2} \\
&\leq \sup_{\theta\in\Theta_{I}}\big\|\mathbb{G}_{N}\,g(Z_{i};\theta,\xi_{0})\big\|_{+}^{2}
= O_{P}(1),
\end{align*}
the last step by the $P$-Donsker property \eqref{eq:pdonsker} and the continuous mapping theorem. 

\medskip
\noindent\emph{Verification of \Cref{ass:cons}(c).}
Again by the $P$-Donsker property,
\begin{align*}
&\sup_{\theta\in\Theta}\sqrt{N}\,\Big|\|\mathrm{E}\,g(Z;\theta,\xi_{0})\|_{+}^{2}-\|\mathrm{E}_{N}\,g(Z;\theta,\xi_{0})\|_{+}^{2}\Big| \\
&\qquad\leq \Big(\sup_{\theta}\|\mathrm{E}\,g(Z;\theta,\xi_{0})\|_{+}+\sup_{\theta}\|\mathrm{E}_{N}\,g(Z;\theta,\xi_{0})\|_{+}\Big)\,\sup_{\theta}\|\mathbb{G}_{N}\,g(Z;\theta,\xi_{0})\|
= O_{P}(1).
\end{align*}

\medskip
\noindent\emph{Verification of \Cref{ass:cons}(d).}
Fix $p>0$. By \Cref{ass:smallbiasvalue}(b), $\Pr(\mathcal{B}_{N})\geq 1-p/3$ for $N$ large, and by \Cref{lem:concentrate} and \eqref{eq:pdonsker} there exists $C_{p/3}>0$ such that, for $N$ large,
\begin{align*}
\Pr(\mathcal{E}_{1,N}') &:= \Pr\left(\sup_{\theta\in\Theta}\Big\|\mathbb{G}_{N}\left\{g(Z_{i};\theta,\widehat{\xi}_{i})-g(Z_{i};\theta,\xi_{0})\right\}\Big\| \;\leq\; C_{p/3}\right) \;\geq\; 1-p/3, \\
\Pr(\mathcal{E}_{2,N}') &:= \Pr\left(\sup_{\theta\in\Theta}\big\|\mathbb{G}_{N}\,g(Z_{i};\theta,\xi_{0})\big\| \;\leq\; C_{p/3}\right) \;\geq\; 1-p/3.
\end{align*}
On $\mathcal{B}_{N}\cap\mathcal{E}_{1,N}'\cap\mathcal{E}_{2,N}'$ and for large $N$, by the fold-conditional convention \eqref{eq:foldE},
\begin{align*}
N\,Q_{N}(\theta,\widehat{\xi})
\;=\; \Big\|\mathbb{G}_{N}\,g(Z_{i};\theta,\widehat{\xi}_{i}) + \sqrt{N}\,\mathrm{E}\,g(Z_{i};\theta,\widehat{\xi}_{i})\Big\|_{+}^{2}.
\end{align*}
Fix $\theta$ with $d(\theta,\Theta_{I})\geq 6C_{p/3}/(C_{\min}\sqrt{N})$ and set
\begin{align*}
\mu_{N} &:= \sqrt{N}\,\mathrm{E}\,g(Z_{i};\theta,\xi_{0}), \\
R_{N} &:= \mathbb{G}_{N}\,g(Z_{i};\theta,\xi_{0}) + \mathbb{G}_{N}\left\{g(Z_{i};\theta,\widehat{\xi}_{i})-g(Z_{i};\theta,\xi_{0})\right\} \\
&\qquad + \sqrt{N}\,\mathrm{E}\left[g(Z_{i};\theta,\widehat{\xi}_{i})-g(Z_{i};\theta,\xi_{0})\right].
\end{align*}
By \Cref{ass:partid},
\begin{align*}
\|\mu_{N}\|_{+} \;\geq\; \sqrt{N}\,C_{\min}\,\big(d(\theta,\Theta_{I})\wedge\delta_{\min}\big) \;\geq\; 6C_{p/3}.
\end{align*}
By $\mathcal{E}_{1,N}'$, $\mathcal{E}_{2,N}'$ and the bound $\sqrt{N}\,\bar{s}_{N}<C_{p/3}$ for $N$ large, which \Cref{ass:smallbiasvalue}(e) gives for the fold-conditional bias on $\mathcal{B}_{N}$,
\begin{align*}
\|R_{N}\| \;\leq\; C_{p/3}+C_{p/3}+C_{p/3} \;=\; 3C_{p/3}.
\end{align*}
Hence $\|\mu_{N}\|_{+}\geq 2\|R_{N}\|$, and the triangle inequality gives
\begin{align*}
\|\mu_{N}+R_{N}\|_{+} \;\geq\; \|\mu_{N}\|_{+}-\|R_{N}\| \;\geq\; \tfrac{1}{2}\,\|\mu_{N}\|_{+},
\end{align*}
so that
\begin{align*}
N\,Q_{N}(\theta,\widehat{\xi}) \;=\; \|\mu_{N}+R_{N}\|_{+}^{2}
\;\geq\; \tfrac{1}{4}\,\|\mu_{N}\|_{+}^{2}
\;\geq\; \tfrac{1}{4}\,N\,C_{\min}^{2}\,[d(\theta,\Theta_{I})\wedge\delta_{\min}]^{2}.
\end{align*}
This verifies \Cref{ass:cons}(d) with $\kappa=\frac{1}{4}C_{\min}^{2}$, $\omega=\delta_{\min}$, and $D_{p}=(6C_{p/3}/C_{\min})^{2}$. \Cref{ass:noeffect}(a) holds by \Cref{lem:rateineq}, so \Cref{lem:rate}(i)--(ii) give containment with probability approaching one and the rate.\end{proof}

\begin{proof}[Proof of \Cref{lem:degeneracy}]
Step~1 shows that the cross-fitted criterion vanishes on a contraction $\Theta_{I}^{-\widehat{\epsilon}_{N}}$ with radius $\widehat{\epsilon}_{N}=O_{P}(N^{-1/2})$, which is the degeneracy property of \Cref{def:degeneracy}. Step~2 derives the rate from it and the choice of contour level $\widehat{c}'$.

\medskip
\noindent\emph{Step 1: degeneracy on an inner contraction.}
Fix $\epsilon\in(0,\delta]$ and $\theta\in\Theta_{I}^{-\epsilon}$. By the fold-conditional convention \eqref{eq:foldE},
\begin{align*}
N\,Q_{N}(\theta,\widehat{\xi}) \;=\; \Big\|\mathbb{G}_{N}\,g(Z_{i};\theta,\widehat{\xi}_{i}) + \sqrt{N}\,\mathrm{E}\,g(Z_{i};\theta,\widehat{\xi}_{i})\Big\|_{+}^{2},
\end{align*}
and componentwise, for $l=1,\dots,L$,
\begin{align*}
\mathbb{G}_{N}\,g_{l}(\theta,\widehat{\xi}) + \sqrt{N}\,\mathrm{E}\,g_{l}(\theta,\widehat{\xi})
 &= \underbrace{\mathbb{G}_{N}\,g_{l}(\theta,\xi_{0}) + \sqrt{N}\,\mathrm{E}\,g_{l}(\theta,\xi_{0})}_{\text{(A)}} \\
 &\quad + \underbrace{\mathbb{G}_{N}\big[g_{l}(\theta,\widehat{\xi})-g_{l}(\theta,\xi_{0})\big]}_{\text{(B)}}
 + \underbrace{\sqrt{N}\,\mathrm{E}\big[g_{l}(\theta,\widehat{\xi})-g_{l}(\theta,\xi_{0})\big]}_{\text{(C)}},
\end{align*}
where (A) is the term at the true nuisance, (B) the empirical process of the nuisance estimation error, and (C) the fold-conditional bias.

By \eqref{eq:deg}, $\mathrm{E}\,g_{l}(\theta,\xi_{0})\leq -C(\epsilon\wedge\delta)=-C\epsilon$ for $\theta\in\Theta_{I}^{-\epsilon}$, and by \eqref{eq:pdonsker} the oracle process is bounded,
\begin{align*}
\mathbb{M}_{N}:=\sup_{\theta\in\Theta,\;1\leq l\leq L}\big|\mathbb{G}_{N}\,g_{l}(Z_{i};\theta,\xi_{0})\big| = O_{P}(1),
\end{align*}
so $\text{(A)}\leq\mathbb{M}_{N}-\sqrt{N}C\epsilon$. For (B), \Cref{lem:concentrate} gives $\sup_{\theta}\max_{l}|\text{(B)}|\lesssim_{P}\zeta_{N}=o_{P}(1)$ by \Cref{ass:concentration}(c). For (C), \Cref{ass:smallbiasvalue}(e) gives $\sup_{\theta}\max_{l}|\text{(C)}|\leq\sqrt{N}\bar{s}_{N}=o(1)$ on $\mathcal{B}_{N}$. Hence the event
\begin{align*}
\mathcal{H}_{N}:=\Big\{\sup_{\theta\in\Theta}\max_{l}\big(|\text{(B)}|+|\text{(C)}|\big)\leq 1\Big\}
\end{align*}
has probability approaching one. Choose the radius
\begin{align*}
\widehat{\epsilon}_{N} \;:=\; \frac{2}{C\sqrt{N}}\,\big(\mathbb{M}_{N}+1\big),
\end{align*}
which is $O_{P}(N^{-1/2})$ and at most $\delta$ with probability approaching one. On $\mathcal{H}_{N}\cap\{\widehat{\epsilon}_{N}\leq\delta\}$, for every $\theta\in\Theta_{I}^{-\widehat{\epsilon}_{N}}$ and every $l$,
\begin{align*}
\mathbb{G}_{N}\,g_{l}(\theta,\widehat{\xi}) + \sqrt{N}\,\mathrm{E}\,g_{l}(\theta,\widehat{\xi}) \;\leq\; \mathbb{M}_{N} - \sqrt{N}\,C\,\widehat{\epsilon}_{N} + 1 \;=\; -(\mathbb{M}_{N}+1) \;\leq\; 0,
\end{align*}
so all positive parts vanish and $Q_{N}(\theta,\widehat{\xi})=0=\inf_{\Theta}Q_{N}(\cdot,\widehat{\xi})$ for all $\theta\in\Theta_{I}^{-\widehat{\epsilon}_{N}}$. Consequently, with probability approaching one, $\Theta_{I}^{-\widehat{\epsilon}_{N}}\subseteq C_{N}(\widehat{c}';\widehat{\xi})$ for any level $\widehat{c}'\geq 0$, and by the second part of \eqref{eq:deg}, $d_{H}(\Theta_{I}^{-\widehat{\epsilon}_{N}},\Theta_{I})\leq M\widehat{\epsilon}_{N}=O_{P}(N^{-1/2})$, which together is \Cref{def:degeneracy} with $\Theta_{N}=\Theta_{I}^{-\widehat{\epsilon}_{N}}$.

\medskip
\noindent\emph{Step 2: rate from degeneracy and choice of $\widehat{c}'$.}
Let $\widehat{c}'=O_{P}(1)$ obey \eqref{eq:minc:deg} and set $\widehat{c}:=\max\{\widehat{c}',\sup_{\theta\in\Theta_{I}}N\,Q_{N}(\theta,\widehat{\xi})\}$, which is $O_{P}(1)$ by \Cref{ass:cons}(b) and \Cref{ass:noeffect}(b), both supplied by \Cref{lem:rateineq} and the proof of \Cref{thm:main:ineq}, and satisfies the level condition of \Cref{thm:main:ineq} by construction. By Step~1 there is $\widehat{\epsilon}_{N}=O_{P}(N^{-1/2})$ such that, with probability approaching one, $Q_{N}(\theta,\widehat{\xi})=0$ on $\Theta_{N}:=\Theta_{I}^{-\widehat{\epsilon}_{N}}$, hence
\begin{align*}
\Theta_{N} \;\subseteq\; C_{N}(\widehat{c}';\widehat{\xi}) \;\subseteq\; C_{N}(\widehat{c};\widehat{\xi}) \qquad\text{with probability approaching one.}
\end{align*}
By the second part of \eqref{eq:deg}, applicable since $\widehat{\epsilon}_{N}\leq\delta$ with probability approaching one, $d_{H}(\Theta_{N},\Theta_{I})\leq M\widehat{\epsilon}_{N}=O_{P}(N^{-1/2})$. By \Cref{thm:main:ineq} and $\widehat{c}=O_{P}(1)$, $d_{H}(C_{N}(\widehat{c};\widehat{\xi}),\Theta_{I})=O_{P}(N^{-1/2})$. The sandwich gives both directions,
\begin{align*}
\sup_{\theta\in C_{N}(\widehat{c}';\widehat{\xi})}d(\theta,\Theta_{I})&\leq\sup_{\theta\in C_{N}(\widehat{c};\widehat{\xi})}d(\theta,\Theta_{I}), \\
\sup_{\theta\in\Theta_{I}}d(\theta,C_{N}(\widehat{c}';\widehat{\xi}))&\leq\sup_{\theta\in\Theta_{I}}d(\theta,\Theta_{N}),
\end{align*}
and both right-hand sides are $O_{P}(N^{-1/2})$. Hence $d_{H}(C_{N}(\widehat{c}';\widehat{\xi}),\Theta_{I})=O_{P}(N^{-1/2})$.\end{proof}

\begin{proof}[Proof of \Cref{thm:main:subs}]
Step~1 brackets each feasible block statistic between two statistics that use $\xi_{0}$ and depend only on the data of their own block, Step~2 shows that the empirical distribution functions of the brackets concentrate, and Step~3 passes to the limit law $\mathcal{T}$ and derives part~(1); part~(2) follows.

\medskip
\noindent\emph{Step 1: two-sided bounds.}
Let $\epsilon_{N}=\sqrt{(1\vee c_{N})/N}$ be the radius of \Cref{ass:fast}, and let $C_{0}:=2\sqrt{2}/C_{\min}$ be the constant with which Step~3 of the proof of \Cref{lem:rate} delivers the rate at the level $c_{N}$: for $N$ large, $\epsilon_{N,p}^{2}=2c_{N}/(\kappa N)$ there, with $\kappa=C_{\min}^{2}/4$, so $\Pr(d_{H}(C_{N}(c_{N};\widehat{\xi}),\Theta_{I})\leq C_{0}\epsilon_{N})\to 1$. \Cref{ass:fast,ass:c5} are stated for every fixed multiple of $\epsilon_{N}$, so they apply at the radius $C_{0}\epsilon_{N}$. For a subsample $j$ (a block of size $b$; see Step~2), define
\begin{align*}
\overline{T}_{j,b,\xi} := \sup_{\theta\in\Theta_{I}^{C_{0}\epsilon_{N}}}b\,Q_{j,b}(\theta,\xi), \qquad
\underline{T}_{j,b,\xi} := \inf_{S\in\mathcal{K}_{N}}\;\sup_{\theta\in S}b\,Q_{j,b}(\theta,\xi),
\end{align*}
where $\mathcal{K}_{N}:=\{S\subset\Theta:\;S\neq\varnothing,\;S\text{ compact},\;d_{H}(S,\Theta_{I})\leq C_{0}\epsilon_{N}\}$. Under the continuity hypothesis, the set $\mathcal{K}_{N}$ is compact under the Hausdorff metric, and since $S\mapsto\sup_{\theta\in S}Q_{j,b}(\theta,\xi)$ is continuous, the infimum is attained at some $\Theta_{b}^{*}\in\mathcal{K}_{N}$; under \eqref{eq:c5unif}, attainment is not used.

Let $\mathcal{A}_{N}:=\{d_{H}(C_{N}(c_{N};\widehat{\xi}),\Theta_{I})\leq C_{0}\epsilon_{N}\}$. The deterministic level $c_{N}\to\infty$ with $c_{N}/N\to 0$ satisfies the level condition of \Cref{lem:rate} by Step~4 of its proof, so $\Pr(\mathcal{A}_{N})\to 1$ as just explained. On $\mathcal{A}_{N}\cap\mathcal{B}_{N}$ we may take $S=C_{N}(c_{N};\widehat{\xi})$ itself, which is compact under the continuity hypothesis, non-empty, and satisfies $d_{H}(S,\Theta_{I})\leq C_{0}\epsilon_{N}$, so $S\in\mathcal{K}_{N}$ and $S\subseteq\Theta_{I}^{C_{0}\epsilon_{N}}$. Hence
\begin{align*}
\underline{T}_{j,b,\widehat{\xi}} \;\leq\; \widehat{\mathcal{T}}_{j,b} := \sup_{\theta\in C_{N}(c_{N};\widehat{\xi})}b\,Q_{j,b}(\theta,\widehat{\xi}) \;\leq\; \overline{T}_{j,b,\widehat{\xi}}.
\end{align*}

Fix $\epsilon>0$ and let $E_{j}$ denote the event $\sup_{\theta\in\Theta_{I}^{C_{0}\epsilon_{N}}}b\,|Q_{j,b}(\theta,\widehat{\xi})-Q_{j,b}(\theta,\xi_{0})|\leq\epsilon$. On $\mathcal{A}_{N}\cap E_{j}$,
\begin{align*}
\underline{T}_{j,b,\xi_{0}}-\epsilon \;\leq\; \underline{T}_{j,b,\widehat{\xi}} \;\leq\; \widehat{\mathcal{T}}_{j,b} \;\leq\; \overline{T}_{j,b,\widehat{\xi}} \;\leq\; \overline{T}_{j,b,\xi_{0}}+\epsilon.
\end{align*}
Because every block draws $b/K$ observations from each fold, the observations of block $j$ in fold $k$ are i.i.d.\ and independent of $\widehat{\xi}^{(k)}$, so \Cref{lem:rateineq} applies to block $j$ with the full-sample nuisance and gives $\Pr(E_{j}^{c})\to 0$ for each $j$, with a bound that depends on $j$ only through $b$. Let $\widehat{\pi}_{N}:=B_{N}^{-1}\sum_{j}\mathbf{1}\{E_{j}^{c}\}$ be the fraction of blocks on which the two-sided bound may fail. Then $\mathrm{E}\,\widehat{\pi}_{N}\to 0$ and $\widehat{\pi}_{N}=o_{P}(1)$ by Markov's inequality. No union bound over the $B_{N}\to\infty$ blocks is needed.

Define the empirical cdfs over the $B_{N}$ subsamples (blocks):
\begin{align*}
\widehat{G}(x) &:= \frac{1}{B_{N}}\sum_{j=1}^{B_{N}}\mathbf{1}\{\widehat{\mathcal{T}}_{j,b}\leq x\}, \\
\underline{G}(x) &:= \frac{1}{B_{N}}\sum_{j=1}^{B_{N}}\mathbf{1}\{\underline{T}_{j,b,\xi_{0}}\leq x\}, \\
\overline{G}(x) &:= \frac{1}{B_{N}}\sum_{j=1}^{B_{N}}\mathbf{1}\{\overline{T}_{j,b,\xi_{0}}\leq x\}.
\end{align*}
From the two-sided bound, since $\widehat{\mathcal{T}}_{j,b}\leq x$ whenever $\overline{T}_{j,b,\xi_{0}}\leq x-\epsilon$ and $\underline{T}_{j,b,\xi_{0}}\leq x+\epsilon$ whenever $\widehat{\mathcal{T}}_{j,b}\leq x$ on $E_{j}$, on $\mathcal{A}_{N}$ and for any fixed $\epsilon>0$ and every $x$,
\begin{align*}
\overline{G}(x-\epsilon)-\widehat{\pi}_{N} \;\leq\; \widehat{G}(x) \;\leq\; \underline{G}(x+\epsilon)+\widehat{\pi}_{N}.
\end{align*}

The compactness of $C_{N}(c_{N};\widehat{\xi})$ and the attainment of suprema used above follow from the continuity hypothesis, under which $\theta\mapsto Q_{j,b}(\theta,\xi)$ is continuous for each fixed $(j,b,\xi)$. Under \eqref{eq:c5unif} instead, the argument runs with the word ``compact'' deleted from the definition of $\mathcal{K}_{N}$, and in Steps~2 and~3 the statistics $\underline{T}_{j,b,\xi_{0}}$ and $\overline{T}_{j,b,\xi_{0}}$ are replaced by the brackets
\begin{equation*}
\big(\sqrt{T_{b,j}}-\Delta_{b,j}\big)_{+}^{2}\;\leq\;\underline{T}_{j,b,\xi_{0}}\;\leq\;\overline{T}_{j,b,\xi_{0}}\;\leq\;\big(\sqrt{T_{b,j}}+\Delta_{b,j}\big)^{2},
\end{equation*}
where $T_{b,j}$ and $\Delta_{b,j}$ are $T_{b}$ and $\Delta_{b}$ of \eqref{eq:c5unif} computed on block $j$. The brackets hold because $\mathcal{K}_{N}\subseteq\mathcal{S}_{N}(C_{0})$ and $\Theta_{I}^{C_{0}\epsilon_{N}}\in\mathcal{S}_{N}(C_{0})$, they are i.i.d.\ across blocks, and both converge in law to $\mathcal{T}$ because $T_{b}\Rightarrow\mathcal{T}$ and $\Delta_{b}\to_{p}0$.

\medskip
\noindent\emph{Step 2: concentration for block subsampling.}
Partition the sample into $B_{N}:=\lfloor N/b\rfloor$ disjoint blocks of size $b$, and let block $j$ define $\underline{T}_{j,b,\xi_{0}}$ and $\overline{T}_{j,b,\xi_{0}}$. Both statistics are functions of block $j$ alone: $\mathcal{K}_{N}$ depends only on $\Theta_{I}$ and on the deterministic radius $C_{0}\epsilon_{N}$, and $\xi_{0}$ is non-random. Since the blocks are disjoint and the data are i.i.d., $(\underline{T}_{j,b,\xi_{0}})_{j\leq B_{N}}$ are i.i.d.\ within the row, and $\underline{G}(x)$ is an average of $B_{N}$ i.i.d.\ Bernoulli indicators. Hence, by Chebyshev's inequality, for each fixed $x$,
\begin{align*}
\big|\underline{G}(x)-\mathrm{E}\,\underline{G}(x)\big| \;=\; O_{P}\!\big(B_{N}^{-1/2}\big) \;=\; O_{P}\!\big(\sqrt{b/N}\big),
\end{align*}
and similarly for $\overline{G}(x)$. The Dvoretzky--Kiefer--Wolfowitz inequality upgrades this to a bound uniform in $x$ of the same order. These deviations vanish under $b/N\to 0$ alone. The \emph{estimated} statistics $\widehat{\mathcal{T}}_{j,b}$ are dependent across blocks, because $C_{N}(c_{N};\widehat{\xi})$ is built from the full sample. The two-sided bound of Step~1 is what lets us work with the independent brackets instead, and it is why the argument does not require the condition $b=o(\sqrt{N})$ that a bounded-differences bound applied directly to $\widehat{G}$ would force.

Moreover, blocks are identically distributed, hence
\begin{align*}
\mathrm{E}\,\underline{G}(x) = \Pr(\underline{T}_{j,b,\xi_{0}}\leq x), \qquad
\mathrm{E}\,\overline{G}(x) = \Pr(\overline{T}_{j,b,\xi_{0}}\leq x).
\end{align*}

\medskip
\noindent\emph{Step 3: small-$b$ limit and quantiles.}
By \Cref{ass:c5}, at every continuity point $x$ of the cdf of $\mathcal{T}$,
\begin{align*}
\Pr(\underline{T}_{j,b,\xi_{0}}\leq x)-\Pr(\mathcal{T}\leq x) \;\to\; 0, \qquad
\Pr(\overline{T}_{j,b,\xi_{0}}\leq x)-\Pr(\mathcal{T}\leq x) \;\to\; 0,
\end{align*}
as $b\to\infty$. Combining with Step~2 gives
\begin{align*}
\underline{G}(x) = \Pr(\mathcal{T}\leq x)+o_{P}(1), \qquad
\overline{G}(x) = \Pr(\mathcal{T}\leq x)+o_{P}(1).
\end{align*}
Using the two-sided bound, $\widehat{\pi}_{N}=o_{P}(1)$ and the continuity of the cdf of $\mathcal{T}$ (\Cref{ass:c4}), for any fixed continuity point $x$ and any fixed $\epsilon>0$,
\begin{align*}
\Pr(\mathcal{T}\leq x-\epsilon)+o_{P}(1) \;\leq\; \widehat{G}(x) \;\leq\; \Pr(\mathcal{T}\leq x+\epsilon)+o_{P}(1).
\end{align*}
Letting $\epsilon\downarrow 0$ yields $\widehat{G}(x)=\Pr(\mathcal{T}\leq x)+o_{P}(1)$ at continuity points. Write $F$ for the cdf of $\mathcal{T}$ and $q:=q_{1-\tau}(\mathcal{T})$. Since $\Pr(\mathcal{T}=0)=F(0)<1-\tau$, $q>0$, so $F$ is continuous at $q$ by \Cref{ass:c4} and $F(q)=1-\tau$. Since $F$ is strictly increasing near $q$, $F(q-\delta)<1-\tau<F(q+\delta)$ for all small $\delta>0$, and $q\pm\delta$ can be taken to be continuity points. Applying the convergence of $\widehat{G}$ at $q\pm\delta$ shows that the $(1-\tau)$-quantile $\widehat{c}_{\tau}$ of $\widehat{G}$ lies in $[q-\delta,q+\delta]$ with probability approaching one, so $\widehat{c}_{\tau}\to_{p}q$. This proves part~(1).

\medskip
\noindent\emph{Part~(2).}
Coverage fails exactly when some $\theta\in\Theta_{I}$ is excluded, i.e.\ on the event $\{\sup_{\theta\in\Theta_{I}}N\,Q_{N}(\theta,\widehat{\xi})>\widehat{c}_{\tau}\}$. Write $\mathcal{C}_{N}:=\sup_{\theta\in\Theta_{I}}N\,Q_{N}(\theta,\xi_{0})$ for the infeasible statistic and $\widehat{\mathcal{C}}_{N}:=\sup_{\theta\in\Theta_{I}}N\,Q_{N}(\theta,\widehat{\xi})$ for its feasible counterpart. By \Cref{ass:noeffect}(b), which \Cref{lem:rateineq} verifies from the primitive conditions, $|\widehat{\mathcal{C}}_{N}-\mathcal{C}_{N}|\leq N\sup_{\theta\in\Theta_{I}}|Q_{N}(\theta,\widehat{\xi})-Q_{N}(\theta,\xi_{0})|=o_{P}(1)$, so $\widehat{\mathcal{C}}_{N}=\mathcal{C}_{N}+o_{P}(1)$. By \Cref{ass:c4}, $\mathcal{C}_{N}\Rightarrow\mathcal{T}$, hence $\widehat{\mathcal{C}}_{N}\Rightarrow\mathcal{T}$ by Slutsky. By part~(1), $\widehat{c}_{\tau}\to_{p}q_{1-\tau}(\mathcal{T})$, which is a continuity point of the cdf of $\mathcal{T}$ by hypothesis. Applying Slutsky's theorem to the pair $(\widehat{\mathcal{C}}_{N},\widehat{c}_{\tau})$ and using continuity of the limiting cdf at $q_{1-\tau}(\mathcal{T})$,
\begin{align*}
\Pr\!\big(\Theta_{I}\subseteq C_{N}(\widehat{c}_{\tau};\widehat{\xi})\big)
\;=\; \Pr\!\big(\widehat{\mathcal{C}}_{N}\leq\widehat{c}_{\tau}\big)
\;\longrightarrow\; \Pr\!\big(\mathcal{T}\leq q_{1-\tau}(\mathcal{T})\big) \;=\; 1-\tau,
\end{align*}
which is the claim. The procedure is asymptotically exact for $\Theta_{I}$ as a whole, and therefore conservative for any individual $\theta\in\Theta_{I}$.
\end{proof}

\subsection{Verification of the general conditions}
\label{sec:verification}

This section proves \Cref{cor:ex1,cor:iqr,cor:ex2}. We verify the conditions of \Cref{sec:conditions}, namely \Crefrange{ass:partid}{ass:concentration}, \Cref{ass:dominance}, and \Cref{ass:c4,ass:c5}; the conclusions are then \Cref{thm:main:ineq}, \Cref{thm:main:subs} and \Cref{lem:degeneracy}. Throughout the asymptotic analysis, the number of thresholds $S$ and the number of directions $L$ are fixed.

\subsubsection{Affine moment systems}
\label{sec:affine}

Suppose $\mathrm{E}\,g(Z;\theta,\xi_0) = A\theta - c$ for a fixed $L\times d_\theta$ matrix $A$ and vector $c$, so that $\Theta_I=\{\theta\in\Theta: A\theta\leq c\}$ is a polyhedron. The linear case of \Cref{sec:example1-grid} and the example of \Cref{sec:example2} are of this form.

\begin{lemma}[Identifiability and majorant for polyhedra]
\label{lem:hoffman}
Let the polyhedron $\{\theta:A\theta\leq c\}$ be non-empty and contained in $\Theta$, so that $\Theta_I=\{\theta\in\Theta:A\theta\leq c\}=\{\theta:A\theta\leq c\}$. Then \Cref{ass:partid} holds with $C_{\min}=1/H(A)$ and any $\delta_{\min}>0$, where $H(A)<\infty$ is the Hoffman constant of $A$, and \Cref{ass:dominance} holds with $C_{\max}=\|A\|$.
\end{lemma}

\begin{proof}
Hoffman's error bound \citep{Hoffman1952} states that for a consistent system $A\theta\leq c$ there is a finite constant $H(A)$, depending on $A$ alone, with $d(\theta,\Theta_I)\leq H(A)\,\|(A\theta-c)_+\|$ for every $\theta$. Since $\|\mathrm{E}\,g(Z;\theta,\xi_0)\|_+=\|(A\theta-c)_+\|$, this is \Cref{ass:partid} with $C_{\min}=1/H(A)$. For the upper bound, let $\theta^\ast$ be the projection of $\theta$ on $\Theta_I$. Then $A\theta^\ast-c\leq 0$, so $\|(A\theta^\ast-c)_+\|=0$, and since $v\mapsto\|v\|_+$ is $1$-Lipschitz,
\begin{equation*}
\|(A\theta-c)_+\| = \big|\,\|(A\theta-c)_+\|-\|(A\theta^\ast-c)_+\|\,\big| \leq \|A(\theta-\theta^\ast)\| \leq \|A\|\,d(\theta,\Theta_I).
\end{equation*}
\end{proof}

\noindent
For the remaining conditions we additionally use that, in these two cases, $g(Z; \theta, \xi)$ is affine in $\theta$ for every fixed $\xi$. The polyhedron is contained in $\Theta$ in both cases. In the linear case, $\ell_{0,k}\geq 0$ and $u_{0,k}\leq 1$ place the polytope inside $\Theta=[0,1]^S$. In the second example, the grid of directions is symmetric about the origin, so the polyhedron cut out by the $L$ half-spaces is bounded, and $\Theta$ is any compact set containing it. For a finite system of affine moment inequalities the polynomial-minorant condition holds automatically, and the constants are computable from $A$. In the scalar case, $S=1$, the system is $\ell_0\leq\theta\leq u_0$ with $A=(-1,1)'$, so $H(A)=1$ and $\|A\|=\sqrt2$, and \Cref{ass:partid} and \Cref{ass:dominance} hold with $C_{\min}=1$ and $C_{\max}=\sqrt2$.

\Cref{ass:donsker} is immediate in the affine case. Write $g(Z;\theta,\xi_0)=\beta(Z)+B(Z)\theta$. The class $\{g(\cdot;\theta,\xi_0):\theta\in\Theta\}$ is contained in a finite-dimensional vector space of measurable functions, hence it is a VC-subgraph of index at most $d_\theta+2$, and it has a square-integrable envelope $|\beta(Z)|+\|B(Z)\|\sup_{\theta\in\Theta}\|\theta\|$ whenever $\beta$ and $B$ are square-integrable and $\Theta$ is bounded. Such a class is $P$-Donsker \citep[Theorem 2.6.7 and Section 2.5]{VaartWellner}. The same two facts give the entropy bound of \Cref{ass:concentration}(b) with $v=O(d_\theta)$.

\Cref{ass:c4,ass:c5} also follow in the affine case. Because $g$ is affine in $\theta$, the Gaussian limit is $\mathbb{G}(\theta)=\Gamma_0+\Gamma_1\theta$ for a Gaussian pair $(\Gamma_0,\Gamma_1)$. The population moment $A\theta-c$ is Lipschitz in $\theta$, so \Cref{lem:lipschitz} gives the convergence $\mathcal{C}_N\Rightarrow\mathcal{T}$, the continuity of the cdf of $\mathcal{T}$ on $(0,\infty)$ and \eqref{eq:c5unif}, hence \Cref{ass:c5} for data-dependent sets. \Cref{lem:lipschitz} also gives the strict monotonicity of the cdf of $\mathcal{T}$ on $(0,\infty)$ that \Cref{thm:main:subs} requires, whenever $\Pr(\mathcal{T}>0)>0$.

Finally, \eqref{eq:deg} holds in the affine case whenever $\Theta_I$ has non-empty interior. If $B(\theta,\epsilon)\subseteq\Theta_I$ then $B(\theta,\epsilon)$ lies in each half-space $\{\theta:A_l\theta\leq c_l\}$, so $\theta$ is at distance at least $\epsilon$ from each bounding hyperplane and $c_l-A_l\theta\geq\epsilon\|A_l\|$ for every $l$ and the first part of \eqref{eq:deg} holds with $C=\min_l\|A_l\|>0$. For the second part, let $B(\theta_c,r)\subset\Theta_I$. For $\epsilon\leq r$, convexity gives $(1-\epsilon/r)\Theta_I+(\epsilon/r)\theta_c\subset\Theta_I^{-\epsilon}$, so $d_H(\Theta_I^{-\epsilon},\Theta_I)\leq M\epsilon$ with $M=\mathrm{diam}(\Theta_I)/r$. \Cref{lem:degeneracy} then applies with $\delta=r$.

\subsubsection{Moment systems that are not affine}
\label{sec:nonaffine}

\Cref{lem:lipschitz} obtains \Cref{ass:c4} and the uniform approximability \eqref{eq:c5unif} for moments whose mean is Lipschitz in $\theta$. This covers the affine systems of \Cref{sec:affine} and moments that are not continuous in $\theta$, such as \eqref{eq:iqrmoments}. Steps~1 and~2 of the proof give a self-contained version of the limit theory that \citet[Section~4]{CHT} develop for moment-condition models. Write $m(\theta):=\mathrm{E}\,g(Z;\theta,\xi_{0})$, $\mathbb{G}_{N}(\theta):=\mathbb{G}_{N}\,g(Z_{i};\theta,\xi_{0})$, and, for $z\in\ell^{\infty}(\Theta)^{L}$,
\begin{equation*}
\psi(z):=\sup_{\theta\in\Theta_{I}}\big\|I(\theta)+z(\theta)\big\|_{+},
\end{equation*}
so that $\mathcal{T}=\psi(\mathbb{G})^{2}$ and $N\,Q_{N}(\theta,\xi_{0})=\|\sqrt{N}\,m(\theta)+\mathbb{G}_{N}(\theta)\|_{+}^{2}$.

\begin{lemma}[Limit law for Lipschitz moments]
\label{lem:lipschitz}
Suppose \Cref{ass:partid,ass:donsker} hold, $m$ is Lipschitz on $\Theta$ with constant $L_{m}$ and let $b = b_N \rightarrow \infty$ satisfy $b(1 \vee c_N) / N \rightarrow 0$. Then \Cref{ass:c4} and \eqref{eq:c5unif} hold, and the cdf of $\mathcal{T}$ is strictly increasing on $(0,\infty)$ whenever $\Pr(\mathcal{T}>0)>0$.
\end{lemma}

\begin{proof}
The proof uses three properties of $v\mapsto\|v\|_{+}$: it is non-decreasing in each coordinate, it is $1$-Lipschitz, and $\|v\|_{+}=\sup\{u'v:\ u\geq 0,\ \|u\|\leq 1\}$. By the first two, $|\psi(z)-\psi(z')|\leq\sup_{\theta\in\Theta}\|z(\theta)-z'(\theta)\|$, so $\psi$ is continuous on $\ell^{\infty}(\Theta)^{L}$. The set $\Theta_{I}$ is non-empty by \Cref{ass:partid} and compact because $m$ is continuous.

\medskip
\noindent\emph{Step 1: convergence of $\mathcal{C}_{N}$.} For $\theta\in\Theta_{I}$, the coordinates with $I_{l}(\theta)=0$ have $\sqrt{N}\,m_{l}(\theta)=0$ and the others have a non-negative positive part, so monotonicity gives $\mathcal{C}_{N}\geq\psi(\mathbb{G}_{N})^{2}$. For $\eta>0$, let $\psi_{\eta}$ be defined as $\psi$ with $I_{l}(\theta)$ replaced by $0$ if $m_{l}(\theta)>-\eta$ and by $-\infty$ otherwise. On the event $\{\sup_{\theta\in\Theta}\|\mathbb{G}_{N}(\theta)\|<\sqrt{N}\eta\}$, whose probability tends to one, every coordinate with $m_{l}(\theta)\leq-\eta$ has $\sqrt{N}\,m_{l}(\theta)+\mathbb{G}_{N,l}(\theta)<0$, and every other coordinate has $\sqrt{N}\,m_{l}(\theta)\leq 0$ on $\Theta_{I}$, so $\mathcal{C}_{N}\leq\psi_{\eta}(\mathbb{G}_{N})^{2}$ on that event. By \Cref{ass:donsker} and the continuous mapping theorem, $\psi(\mathbb{G}_{N})\Rightarrow\psi(\mathbb{G})$ and $\psi_{\eta}(\mathbb{G}_{N})\Rightarrow\psi_{\eta}(\mathbb{G})$.

As $\eta\downarrow 0$, $\psi_{\eta}(z)$ decreases to $\psi(z)$ for every $z$ with continuous paths. To see this, take $\eta_{n}\downarrow 0$ and $\theta_{n}\in\Theta_{I}$ that attain $\psi_{\eta_{n}}(z)$ up to $1/n$. Along a subsequence, $\theta_{n}\to\theta^{\ast}\in\Theta_{I}$ and the set $\mathcal{L}$ of coordinates with $m_{l}(\theta_{n})>-\eta_{n}$ does not depend on $n$. For $l\in\mathcal{L}$, continuity of $m$ gives $m_{l}(\theta^{\ast})\geq 0$, hence $m_{l}(\theta^{\ast})=0$ and $I_{l}(\theta^{\ast})=0$, and continuity of $z$ gives $\lim_{n}\psi_{\eta_{n}}(z)\leq\|z_{\mathcal{L}}(\theta^{\ast})\|_{+}\leq\psi(z)$, where $z_{\mathcal{L}}$ keeps the coordinates in $\mathcal{L}$. Since $\mathbb{G}$ has continuous paths, $\Pr(\psi_{\eta}(\mathbb{G})^{2}\geq x)\downarrow\Pr(\mathcal{T}\geq x)$ as $\eta\downarrow 0$. At a continuity point $x$ of the cdf of $\mathcal{T}$, the portmanteau theorem and the two bounds give
\begin{equation*}
\Pr(\mathcal{T}>x)\;\leq\;\liminf_{N}\Pr(\mathcal{C}_{N}>x)\;\leq\;\limsup_{N}\Pr(\mathcal{C}_{N}>x)\;\leq\;\Pr\big(\psi_{\eta}(\mathbb{G})^{2}\geq x\big)
\end{equation*}
for every $\eta>0$, and letting $\eta\downarrow 0$ gives $\mathcal{C}_{N}\Rightarrow\mathcal{T}$.

\medskip
\noindent\emph{Step 2: the cdf of $\mathcal{T}$.} By the third property, $\psi(\mathbb{G})=\sup\{u'\mathbb{G}(\theta):(\theta,u)\in\mathcal{U}\}$, where $\mathcal{U}$ collects the pairs with $\theta\in\Theta_{I}$, $u\geq 0$, $\|u\|\leq 1$ and $u_{l}=0$ whenever $I_{l}(\theta)=-\infty$. The set $\mathcal{U}$ is compact because $m$ is continuous, so $\psi(\mathbb{G})$ is the supremum of an almost surely bounded, separable, centred Gaussian process, and it is non-negative because $u=0$ is allowed. The distribution of such a supremum has no atom except possibly at the left end of its support \citep{Tsirelson1975}. Since $\mathbb{G}$ has continuous paths on the compact set $\Theta$, its law is a centred Gaussian measure on the separable space $C(\Theta)^{L}$, whose support is a closed linear subspace and therefore contains the origin. The left end of the support of $\psi(\mathbb{G})$ is therefore $0$, because $\psi$ is continuous with $\psi(0)=0$, and the cdf of $\mathcal{T}$ is continuous on $(0,\infty)$. For strict monotonicity, suppose $\Pr(\mathcal{T}>0)>0$. The support of the law of $\mathbb{G}$ then contains some $z_{0}$ with $\psi(z_{0})>0$, and it contains $tz_{0}$ for every $t\geq 0$ because it is a linear space. Since $\psi(tz_{0})=t\,\psi(z_{0})$, for all $x,h>0$ the open set $\{z:x<\psi(z)^{2}<x+h\}$ contains a point of the support and therefore has positive probability.

\medskip
\noindent\emph{Step 3: uniform approximability.} The observations of a block are i.i.d., so Step~1 at sample size $b$ gives $T_{b}\Rightarrow\mathcal{T}$. Let $\mathbb{G}_{b}$ be the empirical process of the block and $\omega_{b}(r):=\sup_{\|\theta-\theta'\|\leq r}\|\mathbb{G}_{b}(\theta)-\mathbb{G}_{b}(\theta')\|$. Because $\Theta$ is compact and the paths of $\mathbb{G}$ are continuous, hence uniformly continuous, \Cref{ass:donsker} makes $\mathbb{G}_{b}$ asymptotically uniformly equicontinuous in the Euclidean metric, so $\omega_{b}(r_{b})\to 0$ in probability for every $r_{b}\to 0$ \citep[Theorem~1.5.7]{VaartWellner}. Let $S\in\mathcal{S}_{N}(C_{0})$. Every $\theta\in S$ has a $\theta'\in\Theta_{I}$ with $\|\theta-\theta'\|\leq 2C_{0}\epsilon_{N}$, and the first two properties of $\|\cdot\|_{+}$ give
\begin{equation*}
\sqrt{b\,Q_{j,b}(\theta,\xi_{0})}\;\leq\;\sqrt{b\,Q_{j,b}(\theta',\xi_{0})}+2\sqrt{b}\,L_{m}C_{0}\epsilon_{N}+\omega_{b}(2C_{0}\epsilon_{N}).
\end{equation*}
The same bound holds with the roles of $S$ and $\Theta_{I}$ exchanged. Hence \eqref{eq:c5unif} holds with $\Delta_{b}:=2\sqrt{b}\,L_{m}C_{0}\epsilon_{N}+\omega_{b}(2C_{0}\epsilon_{N})$, the same function of the data of each block, and $\Delta_{b}=o_{P}(1)$ because $b\,\epsilon_{N}^{2}=b(1\vee c_{N})/N\to 0$ and $\omega_{b}(2C_{0}\epsilon_{N})=o_{P}(1)$.
\end{proof}

\subsubsection{Proof of \Cref{cor:ex1}}
\label{sec:verif-grid}

In this example, $A$ collects the $2S$ coordinate moment functions \eqref{eq:coordbounds} and the $S-1$ increment moment functions \eqref{eq:increment}, and $\theta\mapsto\mathrm{E}\,g(Z; \theta, \xi)$ is affine with $B(Z)$ constant. \Cref{lem:hoffman} and the preceding paragraphs give \Cref{ass:partid}, \ref{ass:donsker}, \ref{ass:dominance} and \Cref{ass:concentration}(b). Let
\begin{equation*}
\begin{aligned}
\Xi_N := \Big\{\xi=\big(p,F^1(w_1\mid\cdot),\ldots,F^1(w_S\mid\cdot)\big):\ & p\geq\underline p,\ \|p-p_0\|_{P,2}\leq g_N^{p},\\
& \max_{j\leq S}\|F^1(w_j\mid\cdot)-F_0^1(w_j\mid\cdot)\|_{P,2}\leq g_N^{F}\Big\}.
\end{aligned}
\end{equation*}
Then \Cref{ass:smallbiasvalue}(a) holds since $\xi_0\in\Xi_N$, (b) holds with $\phi_N\to0$ by \Cref{ass:ex1ovl,ass:ex1rates}, and (c) holds with $g_N:=g_N^{p}+\sqrt{S}\,g_N^{F}=o(1)$.

\paragraph{Orthogonality and the second-order remainder.}
Only the $S$ lower-bound moment functions involve the nuisance $\xi=\big(p, F ^1(w_1\mid \cdot),\ldots,F ^1(w_S\mid \cdot)\big)$. Fix a grid point $w$ and write $\xi_0=(p_0,F_0 ^1 (w \mid \cdot))$ for the true nuisance, and define deviations $\delta_p (x):= p(x) - p_0(x)$ and $\delta_F(x) := F^1(w \mid x)-F^1_0(w \mid x)$. Because $\mathrm{E}[D\mid X]=p_0(X)$ and $\mathrm{E}[D\mathbf 1\{W\leq w\}\mid X]=p_0(X)F_0 ^1(w\mid X)$,
\begin{equation*}
m(\xi) := \mathrm{E}\,\phi_L(Z;w,\xi) = \mathrm{E}\left[F^1 (w \mid X) + \frac{p_0(X)}{p(X)}\big(F_0 ^1 (w \mid X)-F^1(w \mid X)\big)\right].
\end{equation*}
Substituting $\xi_r=\xi_0+r(\xi-\xi_0)$ and rearranging gives
\begin{equation}\label{eq:exactremainder}
m(\xi_r) - m(\xi_0) \;=\; r^{2}\,\mathrm{E}\!\left[\frac{\delta_F(X)\,\delta_p(X)}{p_0(X)+r\,\delta_p(X)}\right] \qquad\text{for every } r .
\end{equation}
\Cref{ass:smallbiasvalue}(d) follows by differentiating \eqref{eq:exactremainder} at $r=0$. \Cref{ass:smallbiasvalue}(e) follows from \eqref{eq:exactremainder} at $r=1$: since $p_0\geq\underline p$ and $p\geq\underline p$ on $\Xi_N$ by \Cref{ass:ex1ovl}, the denominator satisfies $p_0+r\delta_p\geq\underline p$, and Cauchy--Schwarz gives
\begin{equation}\label{eq:sbar-explicit}
|m(\xi)-m(\xi_0)| \;\leq\; \underline p^{-1}\,\|\delta_p(X)\|_{P,2}\,\|\delta_F(X)\|_{P,2},
\end{equation}
so $\bar s_N$ can be taken as the product of the two first-stage rates,
\begin{equation}\label{eq:ml_rate}
    \bar{s}_{N} \;=\; \underline p^{-1}\,g_N^{p}\,g_N^{F} \;=\; o(N^{-1/2}).
\end{equation}

\paragraph{Envelope and continuity in the nuisance.}
Under the overlap condition, $|\phi_L|\leq 1+\underline p^{-1}$ and $|\phi_U|\leq 1$, so the envelope of \Cref{ass:concentration}(a) may be taken constant with a finite $c$-norm for every $c$. For \Cref{ass:concentration}(c), $\phi_L$ is Lipschitz in $\xi$, and $|\phi_L(\xi)-\phi_L(\xi_0)|\leq (1+\underline p^{-1})|\delta_F(X)| + \underline p^{-2}|\delta_p(X)|$, so $r_N'\leq (1+\underline p^{-1}+\underline p^{-2})\,g_N$. Then $r_N'\sqrt{\log N}=o(1)$ whenever $g_N=o((\log N)^{-1/2})$, which is attained by any polynomial rate. 

Therefore, all the conditions of \Cref{sec:conditions} hold, and \Cref{thm:main:ineq}, \Cref{thm:main:subs} and \Cref{lem:degeneracy} deliver conclusions (1), (2) and (3) of \Cref{cor:ex1}.

\subsubsection{Proof of \Cref{cor:iqr}}
\label{sec:verif-iqr}

For $\theta\in\Theta$ the population moments are
\begin{equation*}
\begin{aligned}
m(\theta)=\Big(&F_0^{\mathrm{lo}}(\theta_1)-\alpha_1,\ F_0^{\mathrm{lo}}(\theta_2)-\alpha_2,\ \alpha_1-F_0^{\mathrm{up}}(\theta_1),\\
&\alpha_2-F_0^{\mathrm{up}}(\theta_2),\ F_0^{\mathrm{up}}(\theta_2)-F_0^{\mathrm{up}}(\theta_1)-(\alpha_2-\alpha_1)\Big)'.
\end{aligned}
\end{equation*}
The densities of $F_0^{\mathrm{lo}}$ and $F_0^{\mathrm{up}}$ are $\mathrm{E}[f_0^1(w\mid X)]$ and $\mathrm{E}[p_0(X)f_0^1(w\mid X)]$. By \Cref{ass:iqrid}(2) both are at most $\bar f$, so $m$ is Lipschitz with constant $L_m\leq 3\bar f$. By \Cref{ass:iqrid}(3), and since $p_0\leq 1$, both are at least $\underline f$ almost everywhere on $\mathcal{J}$, so $|F(t)-F(t')|\geq\underline f|t-t'|$ for $t,t'\in\mathcal{J}$ and $F\in\{F_0^{\mathrm{lo}},F_0^{\mathrm{up}}\}$. The set $\Theta_I$ contains $\theta_0$ by the derivation of \eqref{eq:iqrset}.

\paragraph{Identifiability and majorant.}
\Cref{ass:dominance} holds with $C_{\max}=L_m$, by the argument in the proof of \Cref{lem:hoffman} with $A\theta-c$ replaced by $m(\theta)$. For \Cref{ass:partid}, let $\Phi(\theta):=\big(F_0^{\mathrm{up}}(\theta_1),F_0^{\mathrm{up}}(\theta_2)\big)$. The density bounds make $\Phi$ one-to-one from $\mathcal{J}^2$ onto $\Phi(\mathcal{J})^2$, with $\|\Phi(\theta)-\Phi(\theta')\|\leq\bar f\|\theta-\theta'\|$ on $\Theta$ and $\|\theta-\theta'\|\leq\underline f^{-1}\|\Phi(\theta)-\Phi(\theta')\|$ on $\mathcal{J}^2$. With $\bar s_j:=F_0^{\mathrm{up}}(\bar\theta_j)$, consider the polygon
\begin{equation*}
\mathcal{R}:=\big\{s\in\mathbb{R}^2:\ \alpha_j\leq s_j\leq\bar s_j,\ j=1,2,\quad s_2-s_1\leq\alpha_2-\alpha_1,\quad s_1\leq s_2\big\}=\{s:\ As\leq c\}.
\end{equation*}
Every $\theta\in\Theta_I$ satisfies $\underline\theta_j\leq\theta_j\leq\bar\theta_j$, so $\Theta_I\subset\mathcal{J}^2$. Since $F_0^{\mathrm{lo}}$ and $F_0^{\mathrm{up}}$ are strictly increasing on $\mathcal{J}$, it follows that $\mathcal{R}\subset\Phi(\mathcal{J})^2$ and $\Theta_I=\{\theta\in\mathcal{J}^2:\Phi(\theta)\in\mathcal{R}\}$.

Fix $\theta\in\Theta$ with $d(\theta,\Theta_I)\leq\bar\delta$, so that $\theta\in\mathcal{J}^2$, and let $s:=\Phi(\theta)$. The positive parts of the rows $s_j\geq\alpha_j$ and $s_2-s_1\leq\alpha_2-\alpha_1$ are those of the third to fifth coordinates of $m(\theta)$, and the row $s_1\leq s_2$ holds because $\theta_1\leq\theta_2$. The density bounds on $\mathcal{J}$ give $(s_j-\bar s_j)_+\leq\bar f(\theta_j-\bar\theta_j)_+\leq(\bar f/\underline f)\big(F_0^{\mathrm{lo}}(\theta_j)-\alpha_j\big)_+$. Hence $\|(As-c)_+\|\leq(\bar f/\underline f)\|m(\theta)\|_+$. Hoffman's bound gives $s^\ast\in\mathcal{R}$ with $\|s-s^\ast\|\leq H(A)\|(As-c)_+\|$, and $\theta^\ast:=\Phi^{-1}(s^\ast)\in\Theta_I$ satisfies $\|\theta-\theta^\ast\|\leq\underline f^{-1}\|s-s^\ast\|$, so $d(\theta,\Theta_I)\leq H(A)\,\bar f\,\underline f^{-2}\,\|m(\theta)\|_+$. On the compact set $\{\theta\in\Theta:d(\theta,\Theta_I)\geq\bar\delta\}$ the continuous function $\theta\mapsto\|m(\theta)\|_+$ is positive and, if the set is non-empty, attains a minimum $\kappa_{\bar\delta}>0$; otherwise set $\kappa_{\bar\delta}:=\bar\delta$. \Cref{ass:partid} therefore holds with $\delta_{\min}=\bar\delta$ and $C_{\min}=\min\{\underline f^2/(H(A)\bar f),\,\kappa_{\bar\delta}/\bar\delta\}$.

\paragraph{Donsker property and entropy.}
For $\xi=(p,F^1)$ with $F^1(\cdot\mid x)$ non-decreasing, the subgraphs of the class $\{x\mapsto F^1(t\mid x):t\in[\underline w,\bar w]\}$ are totally ordered by inclusion, so the class is VC-subgraph with index $2$. The same holds for $\{D\,\mathbf 1\{W\leq t\}:t\in\mathbb{R}\}$, which depends on $W$ only through $DW$, and $\{D\,\mathbf 1\{t_1<W\leq t_2\}:t_1\leq t_2\}$ is VC-subgraph with index at most $3$. Each coordinate of $g(\cdot;\theta,\xi)$ is a fixed bounded function plus at most two such functions, each multiplied by one of the fixed functions $1-D/p(X)$, $1/p(X)$ and $1$, which are bounded by $1+\underline p^{-1}$. Multiplication by a fixed function preserves the VC-subgraph property \citep[Lemma~2.6.18]{VaartWellner}, and the uniform covering number of a sum of two classes is at most the product of theirs, so \Cref{ass:concentration}(b) holds with constants that depend only on $\underline p$, and the envelope is the constant $2+\underline p^{-1}$. A class with a bounded envelope and a uniform entropy bound of the form \eqref{eq:entropy} is $P$-Donsker \citep[Theorem~2.5.2]{VaartWellner}. By \Cref{ass:iqrid}(2), $\mathrm{E}|F_0^1(t\mid X)-F_0^1(t'\mid X)|\leq\bar f|t-t'|$ and $\Pr(D=1,\ t<W\leq t')\leq\bar f|t-t'|$ for $t<t'$, so $\mathrm{E}\|g(Z;\theta,\xi_0)-g(Z;\theta',\xi_0)\|^2\leq C\|\theta-\theta'\|$ with $C$ depending on $(\underline p,\bar f)$. The Gaussian limit has paths that are uniformly continuous in this $L^2$ semimetric \citep[Section~1.5]{VaartWellner}, and hence continuous in $\theta$, which completes \Cref{ass:donsker}.

\paragraph{First stage.}
Measure the nuisance error by $\|p-p_0\|_{P,2}+\sup_{w}\|F^1(w\mid\cdot)-F_0^1(w\mid\cdot)\|_{P,2}$, with the supremum over $[\underline w,\bar w]$, and let
\begin{equation*}
\begin{aligned}
\Xi_N := \Big\{\xi=(p,F^1):\ & p\geq\underline p,\ F^1(\cdot\mid x)\ \text{non-decreasing with values in }[0,1]\ \forall x,\\
& \|p-p_0\|_{P,2}\leq g_N^{p},\ \sup_{w}\|F^1(w\mid\cdot)-F_0^1(w\mid\cdot)\|_{P,2}\leq g_N^{F}\Big\}.
\end{aligned}
\end{equation*}
Then \Cref{ass:smallbiasvalue}(a)--(c) hold as in the proof of \Cref{cor:ex1}, now by \Cref{ass:ex1ovl,ass:iqrrates}, with $g_N:=g_N^{p}+g_N^{F}$. Only the two lower-bound moments involve $\xi$, and they evaluate $\phi_L$ at $w=\theta_j$. The identity \eqref{eq:exactremainder} holds at every threshold $w$, so \Cref{ass:smallbiasvalue}(d) holds at every $\theta$, and \eqref{eq:sbar-explicit} at $w=\theta_j$, with $\|F^1(w\mid\cdot)-F_0^1(w\mid\cdot)\|_{P,2}\leq g_N^{F}$ for every $w$, gives \Cref{ass:smallbiasvalue}(e) with $\bar s_N\leq\sqrt2\,\underline p^{-1}g_N^{p}g_N^{F}=o(N^{-1/2})$, uniformly in $\theta\in\Theta$. The envelope and continuity bounds for $\phi_L$ in the proof of \Cref{cor:ex1} hold at every threshold, so \Cref{ass:concentration}(a) and (c) hold with $r_N'\leq\sqrt2\,(1+\underline p^{-1}+\underline p^{-2})\,g_N$.

\paragraph{Limit law and degeneracy.}
\Cref{lem:lipschitz} gives \Cref{ass:c4} and \eqref{eq:c5unif}, which replaces the continuity hypothesis and \Cref{ass:c5} in \Cref{thm:main:subs}. For conclusion~(3), suppose $\Theta_I$ has non-empty interior; then so does $\mathcal{R}$, which contains a ball $B(s_c,r)$. Let $\delta':=\min\{\bar\delta,r/\bar f\}$, let $0<\epsilon\leq\delta'$ and let $B(\theta,\epsilon)\subseteq\Theta_I$. The points $\theta\pm\epsilon e_j$ and $\theta+(\epsilon/\sqrt2)(-1,1)'$ lie in $\Theta_I$, so $\theta_j+\epsilon\leq\bar\theta_j$, $\theta_j-\epsilon\geq\underline\theta_j$ and $F_0^{\mathrm{up}}(\theta_2+\epsilon/\sqrt2)-F_0^{\mathrm{up}}(\theta_1-\epsilon/\sqrt2)\leq\alpha_2-\alpha_1$, and the density bounds on $\mathcal{J}$ give $m_l(\theta)\leq-\underline f\epsilon$ for every $l$. For the second part of \eqref{eq:deg}, let $\epsilon':=\bar f\epsilon\leq r$. If $B(s',\epsilon')\subseteq\mathcal{R}$, then $\theta':=\Phi^{-1}(s')\in\Theta_I$ satisfies $B(\theta',\epsilon)\subseteq\mathcal{J}^2$ and $\Phi(B(\theta',\epsilon))\subseteq B(s',\epsilon')$, hence $B(\theta',\epsilon)\subseteq\Theta_I$; in particular $\Phi^{-1}(s_c)\in\Theta_I^{-\epsilon}$. The convexity argument of \Cref{sec:affine} places every point of $\mathcal{R}$ within $\epsilon'\,\mathrm{diam}(\mathcal{R})/r$ of such an $s'$, and pulling back with $\Phi^{-1}$ gives $d_H(\Theta_I^{-\epsilon},\Theta_I)\leq\bar f\,\mathrm{diam}(\mathcal{R})\,\epsilon/(\underline f\,r)$. \Cref{lem:degeneracy} then applies with $\delta=\delta'$, $C=\underline f$ and $M=\bar f\,\mathrm{diam}(\mathcal{R})/(\underline f\,r)$.

All conditions of \Cref{sec:conditions} therefore hold, with \eqref{eq:c5unif} in place of continuity in $\theta$ and \Cref{ass:c5}, and \Cref{thm:main:ineq}, \Cref{thm:main:subs} and \Cref{lem:degeneracy} deliver conclusions (1), (2) and (3) of \Cref{cor:iqr}.

\subsubsection{Proof of \Cref{cor:ex2}}
\label{sec:verif-emp}

Write $V_0:=V(\eta_0)$ and $\bar\theta:=\sup_{\theta\in\Theta}\|\theta\|$. At the true nuisance, $\mathrm{E}\,g_l(Z;\theta,\xi_0)=q_l'\Sigma_0\theta-c_l$, where $c_l$ is the right-hand side of \eqref{eq:ex2finite_set}, so the system is affine with the $L\times d$ matrix $A$ whose rows are $q_l'\Sigma_0$. Since $\Sigma_0$ is non-singular by \Cref{ass:ex2id}(1) and the symmetric grid spans $\mathbb{R}^d$ by \Cref{ass:ex2id}(2), the rows $q_l'\Sigma_0$ positively span $\mathbb{R}^d$, so the polyhedron $\{\theta:A\theta\leq c\}$ is bounded. It is non-empty, since it contains $\theta_0 \in \Theta_I^{\ast}$, and $\Theta$ contains it by \Cref{ass:ex2id}(3), so \Cref{lem:hoffman} applies. For each fixed $\xi$, $g(Z;\theta,\xi)$ is affine in $\theta$ with bounded coefficients, so \Cref{sec:affine} gives \Cref{ass:partid,ass:donsker,ass:dominance}, \Cref{ass:concentration}(b), \Cref{ass:c4,ass:c5}, and \eqref{eq:deg} whenever $\Theta_I$ has non-empty interior. Let
\begin{equation*}
\begin{aligned}
    \Xi_N := \Big\{\xi=\big(\eta, \ell, \pi_1, \ldots, \pi_L\big):\ & \|\eta(X)\|\leq\bar D,\ |\ell(X)|\leq\bar Y,\ 0\leq\pi_l(X)\leq\bar\Delta\ \text{a.s.},\\
    & \|\eta-\eta_0\|_{P,2}\leq g_{\eta,N},\ \|\ell - \ell_0\|_{P, 2} \leq g_{\ell, N},\\
    & \max_{l\leq L}\|\pi_l - \pi_{0, l}\|_{P,2}\leq g_{\pi, N}\Big\}.
\end{aligned}
\end{equation*}
By \Cref{ass:ex2width}, $\eta_0=\mathrm{E}[D\mid X]$, $\ell_0=\mathrm{E}[Y_L\mid X]$ and $\pi_{0,l}=\mathrm{E}[\Delta\,\mathbf{1}\{q_l'V_0>0\}\mid X]$ satisfy the same bounds, so \Cref{ass:smallbiasvalue}(a) holds; (b) holds by \Cref{ass:ex2rates}; and (c) holds with $g_N := g_{\eta, N} + g_{\ell, N} + \sqrt{L}\, g_{\pi, N}=o(1)$.

\paragraph{Orthogonality and the second-order remainder.}
For $\xi\in\Xi_N$ write $\delta:=\eta-\eta_0$, $\lambda:=\ell-\ell_0$ and $\nu_l:=\pi_l-\pi_{0,l}$, so that $V(\eta)=V_0-\delta(X)$, and let $m_l(\theta,\xi):=\mathrm{E}\,g_l(Z;\theta,\xi)$. Conditioning on $X$ and using $\mathrm{E}[V_0\mid X]=0$ and $\mathrm{E}[Y_L-\ell_0(X)\mid X]=0$, the first two terms of \eqref{eq:emp-moment} contribute $\mathrm{E}[(q_l'\delta)(\delta'\theta)]-\mathrm{E}[(q_l'\delta)\lambda]$ to $m_l(\theta,\xi)-m_l(\theta,\xi_0)$. For the bracketed term, $\mathrm{E}[(q_l'V(\eta))\pi_l]=-\mathrm{E}[(q_l'\delta)\pi_{0,l}]-\mathrm{E}[(q_l'\delta)\nu_l]$, and the definition of $\mathcal{K}_l$ together with $\mathrm{E}[\Delta\,\mathbf{1}\{q_l'V_0>0\}\mid X]=\pi_{0,l}(X)$ gives $\mathrm{E}[\Delta(q_l'V(\eta))_+]-\mathrm{E}[\Delta(q_l'V_0)_+]=\mathcal{K}_l(\eta)-\mathrm{E}[(q_l'\delta)\pi_{0,l}]$. Hence
\begin{equation}\label{eq:exactremainder2}
\begin{aligned}
    m_l(\theta,\xi) - m_l(\theta,\xi_0) &= \mathrm{E}\big[(q_l'\delta(X))(\delta(X)'\theta)\big] - \mathrm{E}\big[(q_l'\delta(X))\lambda(X)\big]\\
    &\quad - \mathrm{E}\big[(q_l'\delta(X))\nu_l(X)\big] - \mathcal{K}_l(\eta).
\end{aligned}
\end{equation}
Along $\xi_r=\xi_0+r(\xi-\xi_0)$ the deviations are $r\delta$, $r\lambda$ and $r\nu_l$, and $\eta_0+r\delta=(1-r)\eta_0+r\eta$ satisfies $\|\eta_0+r\delta\|\leq\bar D$, so \Cref{ass:ex2kink} gives $|\mathcal{K}_l(\eta_0+r\delta)|\leq C_K r^2\|\delta\|_{P,2}^2$. Every term of \eqref{eq:exactremainder2} along the path is therefore $O(r^2)$, and differentiating at $r=0$ gives \Cref{ass:smallbiasvalue}(d). At $r=1$, the Cauchy--Schwarz inequality, $\|q_l\|=1$ and \Cref{ass:ex2kink} give
\begin{equation*}
    \sup_{\theta \in \Theta} \big|m_l(\theta,\xi) - m_l(\theta,\xi_0)\big| \;\leq\; (\bar\theta+C_K)\|\delta\|_{P, 2}^2 + \|\delta\|_{P, 2} \|\lambda\|_{P, 2} + \|\delta\|_{P, 2}\|\nu_l\|_{P, 2},
\end{equation*}
so \Cref{ass:smallbiasvalue}(e) holds with $\bar{s}_N \lesssim g_{\eta, N} (g_{\eta, N} + g_{\ell, N} + g_{\pi, N}) = o(N^{-1/2})$ by \Cref{ass:ex2rates}.

\paragraph{Envelope and continuity in the nuisance.}
For $\xi\in\Xi_N\cup\{\xi_0\}$, \Cref{ass:ex2width} and the bounds in $\Xi_N$ give $\|V(\eta)\|\leq 2\bar D$, $|Y_L-\ell(X)|\leq 2\bar Y$, $0\leq\pi_l\leq\bar\Delta$ and $\Delta\leq\bar\Delta$, so every coordinate of $g$ is bounded by a constant that depends only on $(\bar D,\bar Y,\bar\Delta,\bar\theta)$, and \Cref{ass:concentration}(a) holds for every $c>2$. Each of the four terms of \eqref{eq:emp-moment} is Lipschitz in $(\delta(X),\lambda(X),\nu_l(X))$ with such a constant, because $a\mapsto a_+$ is $1$-Lipschitz. Hence
\begin{equation*}
    \sup_{\theta\in\Theta}\big(\mathrm{E}\|g(Z;\theta,\xi)-g(Z;\theta,\xi_0)\|^2\big)^{1/2} \;\lesssim\; \|\delta\|_{P,2}+\|\lambda\|_{P,2}+\max_{l\leq L}\|\nu_l\|_{P,2},
\end{equation*}
so \Cref{ass:concentration}(c) holds with $r'_N \lesssim g_{\eta, N} + g_{\ell, N} + g_{\pi, N}$, and $r_N'\sqrt{\log N}=o(1)$ by \Cref{ass:ex2rates}.

All the conditions of \Cref{sec:conditions} therefore hold, and \Cref{thm:main:ineq}, \Cref{thm:main:subs} and \Cref{lem:degeneracy} deliver conclusions (1), (2) and (3) of \Cref{cor:ex2}.

\bibliographystyle{chicago}
\bibliography{my_new_bibtex}

\end{document}